%% file: arxiv.tex
\documentclass[11pt]{article}

\usepackage{snapshot}
\usepackage{fullpage}
\usepackage{hyperref}
\usepackage{subcaption}
\usepackage{booktabs}

\usepackage{amsmath,amsthm,amsbsy}
\usepackage{amssymb}
\usepackage[noend]{algorithmic}
\usepackage{mathtools}
\usepackage{microtype}
\usepackage{setspace}
\usepackage{color}
\usepackage{xcolor}
\usepackage{multirow}
\usepackage{paralist}
\usepackage{xspace}
\usepackage[ruled, lined, linesnumbered, commentsnumbered, longend]{algorithm2e}
\usepackage{tikz}
\usetikzlibrary{
	graphs,
	graphs.standard
}
\usepackage{tabularx,booktabs}
\usepackage{verbatim} 
\usepackage{graphicx}
\usepackage{float} 
\usepackage{booktabs} 
\usepackage{thmtools}
\usepackage{thm-restate}

\SetKwInput{KwInput}{Input}     

\allowdisplaybreaks

\hypersetup{hidelinks}

\newcommand{\poly}{\textrm{poly}}
\newcommand{\row}{\textrm{row}}

\newtheorem{theorem}{Theorem}[section]

\newtheorem{lemma}[theorem]{Lemma}
\newtheorem{corollary}[theorem]{Corollary}

\theoremstyle{definition}
\newtheorem{definition}[theorem]{Definition}

\title{Massively Parallel Algorithms for the Stochastic Block Model} 

\newcommand*\samethanks[1][\value{footnote}]{\footnotemark[#1]}
\author{Zelin Li\thanks{School of Computer Science and Technology, University of Science and Technology of China, China. Supported in part by NSFC grant 62272431 and ``the Fundamental Research Funds for the Central Universities''. Email: \tt{meguru@mail.ustc.edu.cn, ppeng@ustc.edu.cn}. } \and Pan Peng\samethanks
\and Xianbin Zhu\thanks{Department of Computer Science, City University of Hong Kong, Hong Kong SAR. Supported in part by a grant from the Research Grants Council of the Hong Kong Special Administrative Region, China [Project No. CityU 11213620]. Email: \tt{xianbin.aaron-zhu@my.cityu.edu.hk}.}}

\begin{document}
	
	\date{}
	\maketitle
	
	\begin{abstract}
Learning the community structure of a large-scale graph is a fundamental problem in machine learning, computer science and statistics. Among others, the Stochastic Block Model (SBM) serves a canonical model for community detection and clustering, and the Massively Parallel Computation (MPC) model is a mathematical abstraction of real-world parallel computing systems, which provides a powerful computational framework for handling large-scale datasets. We study the problem of exactly recovering the communities in a graph generated from the SBM in the MPC model. Specifically, given $kn$ vertices that are partitioned into $k$ equal-sized clusters (i.e., each has size $n$),
a graph on these $kn$ vertices is randomly generated such that each pair of vertices is connected with probability~$p$ if they are in the same cluster and with probability $q$ if not, where $p > q > 0$. 

We give MPC algorithms for the SBM in the (very general) \emph{$s$-space MPC model}, where each machine is guaranteed to have memory $s=\Omega(\log n)$. Under the condition that\footnote{$\tilde{\Omega}(\cdot)$ hides $\poly(\log kn)$ factors.} 
$\frac{p-q}{\sqrt{p}}\geq \tilde{\Omega}(k^{\frac12}n^{-\frac12+\frac{1}{2(r-1)}})$ for any integer $r\in [3,O(\log n)]$, our first algorithm exactly recovers all the $k$ clusters in $O(kr\log_s n)$ rounds using $\tilde{O}(m)$ total space, or in $O(r\log_s n)$ rounds using $\tilde{O}(km)$ total space. If
$\frac{p-q}{\sqrt{p}}\geq \tilde{\Omega}(k^{\frac34}n^{-\frac14})$, our second algorithm achieves $O(\log_s n)$ rounds and $\tilde{O}(m)$ total space complexity. Both algorithms significantly improve upon a recent result of Cohen-Addad et al. [PODC'22], who gave  algorithms that only work in the \emph{sublinear space MPC model}, where each machine has local memory~$s=O(n^{\delta})$ for some constant $\delta>0$, with a much stronger condition on $p,q,k$. 
Our algorithms are based on collecting the $r$-step neighborhood of each vertex and comparing the difference of some statistical information generated from the local neighborhoods for each pair of vertices. To implement the clustering algorithms in parallel, we present efficient approaches for implementing some basic graph operations in the $s$-space MPC model.

\end{abstract}

\section{Introduction}

Graph clustering is a fundamental task in machine learning, computer science and statistics. In this task, given a graph that may represent a social/information/biological network, the goal is to partition its vertex set into a few maximal subsets (called \emph{clusters} or \emph{communities}) of similar vertices. Depending on the context, a cluster may correspond to a social group of people with the same hobbies, a group of web-pages with similar contents or a set of proteins that interact very frequently. Intuitively, in a good clustering of a graph, there are few edges between different clusters while there are relatively many edges inside each cluster. There is no unified formalization on the notions of graph clustering and clusters. Here we focus on a natural and widely-used model for graph clustering, the \emph{stochastic block model (SBM)}. In the SBM, we are given a set $V$ of $N=kn$ vertices such that there is a hidden partition of $V$ with~$V=\bigcup_{i=1}^k V_i$, $V_i\cap V_j=\emptyset$ for any $1\leq i<j\leq k$, where each set $V_i$ is called a \emph{cluster} (or \emph{community}). For simplicity, we assume that each cluster has an equal size, i.e., $|V_i|=n$. We say a graph $G=(V,E)$ is generated from the SBM with parameters $n,p,q,k$, abbreviated as SBM($n,p,q,k$), if for any two vertices $u,v$ that belong to the same cluster, the edge $(u,v)$ appears in $G$ with probability $p$; for any two vertices $u,v$ that belong to two different clusters, the edge $(u,v)$ appears with probability $q$, where $0<q<p<1$.

Thanks to its simplicity and its ability in explaining the community structures in real world data, the SBM has been extensively studied in the computer science literature. Most previous work has been focusing on algorithms that work on a single machine, with the goal of extracting the communities with the optimal (computational and/or statistical) trade-offs between parameters $n,p,q,k$, for different types of recoveries (i.e., exact, weak, and partial recovery). Significant progress has been made on such algorithms (and their limitations) in the past decades (see the survey \cite{Abbe18}). However, most of these algorithms are essentially sequential and cannot be adapted to the parallel or distributed environment, which is unsatisfactory as modern graphs are becoming massive and most of them cannot be fitted into the main memory of a single machine.

We study the problem of exactly recovering communities of a graph from the SBM in the massively parallel computation (MPC) model \cite{KarloffSV10,GoodrichSZ11,BeameKS13}, which is a mathematical abstraction of
modern frameworks of real-world parallel computing systems like MapReduce \cite{DeanG08}, Hadoop \cite{Hadoop}, Spark \cite{ZahariaCFSS10} and Dryad \cite{IsardBYBF07}. In this model, there are $M$ machines that communicate in synchronous rounds, where the local memory of each machine is limited to $s$ words, each of $O(\log n)$ bits. A word is enough to store a node
or a machine identifier from a polynomial (in $n$) domain. Communication is the largest bottleneck in the MPC model. Take the graph problem as an example. The edges of the input graph are arbitrarily distributed across the $M$ machines initially. Ideally, we would like to use minimal number of rounds of computation while using small (say sublinear) space per machine and small total space (i.e., the sum of space used by all machines). 

Recently, Cohen-Addad et al. \cite{Cohen-AddadMS22} gave two algorithms \textsc{Majority}
and \textsc{Louvain} that recover the communities in a graph generated from SBM($n,p,q,k$) when $\frac{p-q}{\sqrt{p}}\geq \Omega(n^{-\frac14+\varepsilon})$ and $k$ is constant. They work in $O(\frac{1}{\varepsilon\cdot \delta})$ rounds in the \emph{sublinear space} MPC model, i.e., each machine has local memory $s=O(n^\delta)$, for any constant $\delta>0$.  Their algorithms and analysis improve upon previous sequential versions of \textsc{Majority} and \textsc{Louvain} given by Cohen-Addad et al \cite{Cohen-AddadKMS20}. Note that for any sequential algorithm, it is known that $\frac{p-q}{\sqrt{p}}=\Omega(\sqrt{\frac{\log n}{n}})$ is necessary for exact recovery even for $k=2$ \cite{AbbeBH16}; there exist spectral algorithms and SDP-based algorithms that find all clusters and achieve this parameter threshold \cite{Abbe18}. Therefore, it is natural to ask \emph{if one can obtain a round-efficient MPC algorithm in the sublinear space model with roughly the same parameter threshold}.

In this paper, we consider a more general setting that we call \emph{the $s$-space MPC model} in which the local memory $s$ is only guaranteed to satisfy that $s=\Omega(\log n)$. Nowadays, the growth rate of data volume far exceeds the growth rate of machine hardware storage and it is likely that we need much more machines to analyze large-scale data. Furthermore,
the problem of clustering of data points from some metric space on such a model has recently received increasing interest \cite{BhaskaraW18,EpastoMZ19,BateniEFM21,Cohen-AddadLNSS21,Cohen-AddadMZ22,coy2023parallel} (see also Section~\ref{sec:relatedwork}), 
partly due to the fact that in some scenarios, the number of clusters $k$ is too large such that even just storing $k$ representatives of all the clusters is not possible in a single machine. Note that this model is more difficult to handle than the sublinear space model, and we need to carefully partition the data across machines so that different machines work in different ``regions of space'' to get a good tradeoff between communication and the used space. Here, we are interested in the question \emph{whether we can obtain an SBM clustering algorithm in the $s$-space MPC model with good tradeoffs between communication, space and SBM parameters}. 

\subsection{Our Results}
We give clustering algorithms for the SBM that work in the $s$-space MPC model where the local memory $s$ of each machine is only guaranteed to satisfy that $s = \Omega(\log n)$. Let $m=\Theta(kn^2p+k^2n^2q)$ denote the total number of edges of the graph (our conditions always imply that $p\geq \Omega(\frac{\log n}{n})$ and $k\leq n$). We use ``with high probability'' to denote ``with probability at least $1-O(n^{-1})$''.

Our first algorithm has the following performance guarantee. 

\begin{restatable}{theorem}{secondmpc}\label{thm:secondmpc}
Let $r$ be any integer such that $3\leq r\leq O(\log n)$. Let $p,q \leq 0.75$ be parameters such that $\max\{p(1-p),q(1-q)\} \geq C_0 \log n/n$ where $C_0>0$ is some constant. 
Suppose that 
$\frac{p-q}{\sqrt{p}} \geq \Omega\left(k^{\frac{1}{2}}n^{-\frac{1}{2}+\frac{1}{2(r-1)}}\log^7 (kn)\right)$. 
Let $G$ be a random graph generated from SBM($n,p,q,k$). Then there exists an algorithm in the $s$-space MPC model that outputs $k$ clusters in $O(kr\log_s n)$ rounds with high probability where each machine has $s=\Omega(\log n)$ memory. The total space used by the algorithm is $\tilde O(m)$.
\end{restatable}
We note that 
the round complexity can be improved to be $O(r\log_s n)$ at the cost of increasing the total space by a $k$ factor, which is formalized in the following theorem. 

\begin{restatable}{theorem}{mpccor}\label{thm:mpc-2}
Under the same condition in Theorem~\ref{thm:secondmpc}, there exists an algorithm that outputs $k$ clusters in $O(r\log_s n)$ rounds where each machine has $s=\Omega(\log n)$ memory with high probability and uses $\tilde O(km)$ total space.
\end{restatable}

Note that for any integer constant $3\leq r\leq o(\log n)$ and any $k\leq \poly(\log n)$, the round complexity of the above algorithm is $O(\log_s n)$ while the total space is $\tilde{O}(m)$. When $r=\Theta(\log n)$, then the recovery condition becomes $\frac{p-q}{\sqrt{p}}\ge \tilde{\Omega}(\sqrt{\frac{{k}}{{n}}})$, which almost matches the statistical limit in the sequential setting up to logarithmic terms \cite{AbbeBH16}. In this case, our algorithm has round complexity $O(\log n\log_s n)$ for any $k\leq \poly(\log n)$ in the $s$-space MPC model. 


When the gap between $p,q$ is sufficiently large, we can achieve $O(\log_s n)$ rounds using $\tilde{O}(m)$ total space, i.e., both the round complexity and the total space complexity are independent of the number $k$ of clusters. Formally, we have the following theorem.
\begin{restatable}{theorem}{firstmpc}\label{thm:firstmpc}
Given a random graph $G$ from SBM($n, p,q,k$) with 
$\frac{p-q}{\sqrt{p}} \geq{\Omega}\left({{k^{\frac34}}{n^{-\frac14}}}(\log n)^{\frac14}\right)$, there exists an algorithm in the $s$-space MPC model that can output $k$ hidden clusters within $O(\log_s n)$ rounds with high probability, where $s = \Omega(\log n)$, and uses $\tilde O(m)$ total space.
\end{restatable}


We note that all the algorithms in Theorem \ref{thm:secondmpc}, \ref{thm:mpc-2} and  \ref{thm:firstmpc} significantly improve the results of \cite{Cohen-AddadMS22}, of which the algorithms only work in the sublinear space MPC model, i.e., $s=O(n^\delta)$ for some constant $\delta>0$, and finish in $O(\frac{1}{\delta\varepsilon})$ rounds, assuming that $\frac{p-q}{\sqrt{p}}\geq n^{-1/4+\varepsilon}$ and $k$ is a constant. 
In both sublinear space and $s$-space models, our algorithms work for a much wider class of SBM graphs (i.e., the requirement on the conditions of $p,q,k$ are much weaker) than those in \cite{Cohen-AddadMS22}. Furthermore, even for the same regime of parameters, our algorithms have better round complexity. For example, in the sublinear space MPC model, our round complexity (from Theorem \ref{thm:firstmpc}) is $O(1/\delta)$ under the condition that $\frac{p-q}{\sqrt{p}} \geq\Omega({n^{-\frac14}}(\log n)^{\frac14})$ and $k$ is constant, while the algorithms in \cite{Cohen-AddadMS22} have round complexity $O(\frac{\log n }{\delta \log \log n})$ under the same condition\footnote{This can be seen by  setting~$\frac{1}{\varepsilon} = \Theta(\frac{\log n}{\log \log n})$ in \cite{Cohen-AddadMS22}.}.
%

Our algorithms are quite different from those in \cite{Cohen-AddadMS22}, in which the algorithms are based on the local-search methods and proceed in rounds by updating the so-called \emph{swap values} for each node to decide where to move the node. Our algorithms are based on collecting the $r$-step neighborhood of each vertex and comparing the difference of some statistical information generated from the local neighborhoods for each pair of vertices.

To implement the above MPC algorithms, we give new algorithms of some basic graphs operations in the $s$-space MPC model in Section~\ref{sec:mbasicmpc}, including \textbf{RandomSet} (for randomly sampling a set), \textbf{ReorganizeNBR} that is for organizing the neighborhood of any two nodes $u,v$ in a set so that they are ``aligned'', i.e., the $i$-th byte of $u$ (or $v$) indicates whether the $i$-th node is the neighbor of $u$ (or $v$ ).  We believe these results will be useful as basic tools in designing algorithms for other problems in the $s$-space MPC model.

\subsection{Our Techniques}

Our MPC algorithms are based on two simple sequential algorithms. We first describe our first algorithm given in Theorem \ref{thm:firstmpc}. It is based on the observation that if $\frac{p-q}{\sqrt{p}} \geq\tilde{\Omega}({{k^{\frac{3}{4}}}{n^{-\frac{1}{4}}}})$, then the number of common neighbors of any two vertices can be used to distinguish if they belong to the same cluster or not. That is, if $u,v$ belong to the same cluster, then the number of their common neighbors is above some threshold $\Delta$; otherwise, the number of common neighbors is smaller than $\Delta$. Let $N(v)$ denote the set of all the neighbors of $v$. 
We further note that to get $k$ clusters of $V$, it is not necessary to compute $|N(u)\cap N(v)|$ for \emph{all pairs} of $u,v$ in $V$, which may cause too much communication for MPC implementation. Instead, we first randomly sample a small set $S'$ with $|S'|=\Theta(k\log n)$. Then we find $k$ representatives of the hidden clusters from $S'$ by computing $|N(u)\cap N(v)|$ for all pairs of $u,v$ in $S'$ and update $S'$ to be the set of $k$ representatives. Then we sample independently another small set $S$ of vertices, and find $k$ sub-clusters from $S$ by computing $|N(u)\cap N(v)|$ for $u\in S$ and $v \in S'$. (A set $T\subseteq S$ is called a \emph{sub-cluster} of some cluster $V_i$ if $ T\subseteq V_i$.)  Based on the $k$ sub-clusters obtained from $S$, we can find all the hidden clusters $V_1,\dots,V_k$ putting any vertex $v\in V\setminus S$ to the sub-cluster that contains the most number of neighbors of $v$.

There are several challenges to implementing the above algorithm in the $s$-space MPC model in which the local memory only satisfies that $s =\Omega(\log n)$. 
Note that in this model, even just to compute the number of common neighbors $|N(u)\cap N(v)|$ for any \emph{fixed pair} $u,v$ in a few parallel rounds (say $O(\log_s n)$ rounds) is non-trivial. The reason is that the neighborhoods $N(u),N(v)$ can be much larger than $s$ and some neighborhoods will be used too many times which leads to large round complexity. To efficiently compute $|N(u)\cap N(v)|$ for $u\in S$ and $v \in S'$, we first show how to reorganize $N(u)$ and $N(v)$ for $u\in S$ and $v\in S'$ so that each byte of $N(u)$ and $N(v)$ for any two nodes aligned; then we can show how to compute $|N(u)\cap N(v)|$ in parallel efficiently by appropriately making some copies of $N(u)$ and $N(v)$. For these tasks, we give detailed MPC implementations of some basic operations, e.g., a procedure for copying neighbors of some carefully chosen nodes and aligning their neighbors while using no more than $\tilde{O}(m)$ total space. 

Our MPC algorithms from Theorem \ref{thm:secondmpc} and \ref{thm:mpc-2} are based on a recent sequential algorithm given in \cite{MZ22SBM}. Roughly speaking, 
one can use the power iterations of some matrix $B = A - q\cdot J$ to find the corresponding clusters, where $A$ is the adjacency matrix of the graph and $J$ is the all-$1$ matrix. It is shown that with high probability, the $\ell_2$-norm of $B_u^r-B_v^r$ is relatively small, if $u,v$ belong to the same cluster; and is large, otherwise. Here $B_u^r$ is the row corresponding to vertex $u$ in the matrix $B^r$. We show that in order to compute $\lVert B_u^r-B_v^r \lVert_2$, it suffices to compute the expressions 
$\mathbf{1}_x^T(A-qJ)^{2r}\mathbf{1}_y$ for all $x,y\in \{u,v\}$.
To do so, we expand the above expression so that we get a sum of terms, each being a vector-matrix-vector multiplication. Then we give a combinatorial explanation of each term, and then calculate it in parallel efficiently based on some basic graph operations in the $s$-space model. 

\subsection{Related Work}\label{sec:relatedwork}
 There is a line of research on metric clustering in the MPC model. In this setting, the input is a set of data points from some metric space (e.g., Euclidean space), and the goal is to find $k$ representative centers, such that some objective function (e.g. the cost functions of $k$-means, $k$-median and $k$-center) is minimized (e.g., \cite{BhaskaraW18}). Bhaskara and Wijewardena \cite{BhaskaraW18} developed an algorithm that outputs $O(k \log k \log n)$ centers whose cost is within a factor of $O((\log n \log \log n)^2)$ of the optimal $k$-means (or $k$-median) clustering, using a memory of $s\in \Omega(d\log n)$ per machine and $O(\log_s n)$ parallel rounds. Note that this does not require $\Omega(k)$ memory per machine. Coy et al. \cite{coy2023parallel} recently improved the approximation ratio of the algorithm for $k$-center in \cite{BhaskaraW18} to $O(\log^* n)$. 
 Cohen-Addad et al. \cite{Cohen-AddadMZ22} gave a fully scalable $(1+\varepsilon)$-approximate $k$-means clustering algorithm when the instance exhibits a ``ground-truth'' clustering structure, captured by a notion of ``$O(\alpha)$-perturbation resilient'', and it uses $O(1)$ rounds and $O_{\varepsilon,d}(n^{1+1/\alpha^2+o(1)})$ total space with arbitrary memory per machine, where each data point is from $\mathbb{R}^d$. 

Regarding the power method for SBM, Wang et al. \cite{WangZS20} proposed an iterative algorithm that first employs the power method with a random starting point and then turns to a generalized power method that can find the communities in a finite number of iterations. Their algorithm runs in nearly linear time and can exactly recover the underlying communities at the information-theoretic limit. Cohen-Addad et al. \cite{Cohen-AddadMS22COLT} further gave a linear-time algorithm that recovers exactly the communities at the asymptotic information-theoretic threshold. Their algorithm is based on similar ideas as in \cite{Cohen-AddadMS22}, that is, given a partition, moving a vertex from one part to the part where it has most neighbors should somewhat improve the quality of the partition. 

Correlation clustering has been studied under the MPC model. In this problem, a signed graph $G=(V,E,\sigma)$ is given as input, and the goal is to partition the vertex set into arbitrarily many clusters so that the disagreement of the corresponding clustering is minimized, where the disagreement is the number of edges that cross different clusters plus the number of non-adjacent pairs inside the clusters \cite{BlellochFS12,KDD-2014-ChierichettiDK,PanPORRJ15,FischerN20,CambusCMU21,Cohen-AddadLMNP21,Assadi022}. The state-of-the-art is a  $(3+\varepsilon)$-approximation algorithm in $O(1/\varepsilon)$ rounds in the massively parallel computation (MPC) with sublinear
space \cite{BehnezhadCMT22}.

\section{Preliminaries}
Consider an undirected graph $G=(V, E)$ where $V$ is the set of vertices, and $E$ is the set of edges. We use $n$ to denote the size of $|V|$ and $m$ to denote the size of $|E|$. Each node in $G$ has a unique ID from $1$ to $n$. We use $\textsf{ID}(u)$ to denote the ID for a node $u\in V$. We use $d(u)$ to denote the degree of $u \in V$. Let $N(u)$ denote the set of neighbors of a node $u\in V$. Given a vertex set $S\subset V$, we use $G[S]$ to denote the subgraph induced by vertices in $S$. In this paper, we abuse the use of node(s) and vertex(vertices). We use $[i]$ to denote $\{1,2,\cdots i\}$. When nodes are active (inactive), they execute (do not execute)  algorithms. 

\bigskip

\noindent\textbf{Chernoff Bound}
Let $X_1,..., X_n$ be independent binary random variables, and $X = \Sigma_{i=1}^n X_i$, and $\mu = \mathbb{E}[X]$. Then it holds that for all $\delta > 0$ that 
$\mathbb{P}[X \geq (1+\delta)\mu] \leq (\frac{e^{-\delta}}{(1+\delta)^{1+\delta}})^{\mu} \leq e^{-\min[\delta^2, \delta]\mu/3}$;
For all $\delta \in (0,1)$, 
$ \mathbb{P}[X \leq (1-\delta)\mu] \leq (\frac{e^{\delta}}{(1-\delta)^{1-\delta}})^{\mu} \leq e^{-\delta^2\mu/2}$.

\bigskip

\noindent\textbf{The MPC model} In this model, we assume that all data is arbitrarily distributed 
among some machines. Let $N$ denote the total amount of data. Each machine has local memory $s$. In our settings, $s = \Omega(\log n)$.  The sum of all local memory is $N = O((m+n)\text{poly}(\log n))$. The communication between any pair of two machines is synchronous, and the bandwidth is $s$ words.

 We ignore the cost of local communication and computation happening in each machine. As the description of the MPC model in the literature, we ignore some communication details among different machines and suppose that all machines are known to each other which means that any machine can send messages to another machine directly (even when the local memory is very small). For the problems in the MPC model, we aim to make the total number of communication rounds among machines as small as possible.

\begin{definition}[separable function]
Let $f: 2^\mathbb{R} \to \mathbb{R}$ denote a set function. We say that $f$ is separable if and only if for any set of reals $A$ and for any $B \subseteq A$, we have $f(A) = f(f(B), f(A\setminus B))$ \footnote{For example, $f$ can be a sum function.}. 
\end{definition}

\begin{lemma}[\cite{behnezhad2019massively}]\label{lem:neigh-mpc}
    Given an $n$-vertex graph, we have~$x_u$ for each node $u\in V$. If the function $f$ is a separable function, then there exists an algorithm that computes $f(\{x_i\in N(u) \})$ for each $u\in V$ with high probability in the sublinear space MPC model in $O(1/\delta)$ rounds using $\tilde O(m)$ space where each machine has space~$O(n^\delta)$.  
\end{lemma}

In the $s$-space MPC model, we restate the following folklore lemma.

\begin{restatable}{lemma}{commsense}\label{lem:comm-sense}
Given an $n$-vertex graph, there exists an algorithm that makes each node $u\in V$ visit $N(u)$ in $O(\log_s n)$ rounds with high probability. 
\end{restatable}

The sorting algorithm is a very important black-box tool in the MPC model, which is stated as follows.
\begin{theorem}[\cite{GoodrichSZ11}]
    Sorting can be solved in $O(\log_s n)$ rounds in the $s$-space MPC model.
\end{theorem}

Furthermore, it has been shown that indexing and prefix-sum operation can be performed in $O(\log_s n)$ rounds \cite{GoodrichSZ11}.
We refer to the \textsc{Index} Algorithm for solving indexing problems in the $s$-space MPC model and the \textsc{Sorting} Algorithm for solving sorting problems in the same model. 
Throughout the following context, we will rely on the fundamental properties associated with the aforementioned operations, as well as Lemma~\ref{lem:neigh-mpc} and Lemma~\ref{lem:comm-sense} by default. 

\section{Implementing Basic Graph Operations in the $s$-space MPC Model}\label{sec:mbasicmpc}

{In this section, we present algorithms for several fundamental graph operations in the $s$-space model, which will be utilized in our MPC algorithms. To the best of our knowledge, most of these operations have not been previously implemented in the $s$-space MPC model. We denote the machines holding node $x$ as $M_x$. (It is important to note that the MPC model follows an edge-partition model, which means that multiple machines may hold the same vertices).  } 
It is worth mentioning that in order to implement some of our proposed algorithms, we utilize previous algorithms for basic MPC operations, as demonstrated in 
Appendix~\ref{sec:appendix}.

\bigskip

\noindent \textbf{\textsc{RandomSet}}~~In the {RandomSet} problem, given an input value $X = \Omega(\log n)$, our goal is to output a random set $S = \{x_1,x_2,\ldots,x_{|S|}\}$  where each element $x_i$ ($i\in[|S|]$) is selected uniformly and independently at random
and $S' = \bigl\{ (x,y) | x \in S\cap M_x, y= \text{Ind}_S(x)  \bigr\} $
where $|S| = \Theta(X)$, 
and $\text{Ind}_S(x)$ is the index of of $x \in S$ in $S$. We use the algorithm~\textsc{RandomSet} to solve the RandomSet problem.

\begin{restatable}{lemma}{randomset}\label{lem:randomset}
The RandomSet problem can be solved in the $s$-space MPC model in $O(\log_s n)$ rounds where $s = \Omega(\log n)$.
\end{restatable}

\begin{proof}
First, we let consecutive machines with a starting index store $v_1,\ldots,v_n$. 
Each node in these consecutive machines is selected uniformly and independently at random. In such a way, we select $\Theta(X)$ nodes (by the Chernoff bound, with high probability there are $\Theta(X)$ nodes being selected if each node has probability of, e.g., $\frac{20X}{n}$ to be selected). For each $v_i \in V$, we compute $d(v_i)$. Recall that we can compute prefix-sum within $O(\log_s n)$ rounds. Therefore, we can obtain a set $I = \{I_1,I_2,\cdots\}$ of numbers $\{d(v_1),d(v_1)+d(v_2),\cdots, \Sigma_{i=1}^{j-1}{d(v_j)} \}$. We store the set $I$ in consecutive machines and starting position is sent to all machines. Then, each node $v_j \in S$ can access $\Sigma_{i=1}^{j-1}{d(v_j)}$ in constant rounds. Next, we copy  the index of $v_j$ in $S$ to all machines with space indexes from $I_{j}$ to $I_{j+1}$ which can be done in $O(\log_s n)$ rounds. Thus, we have a $\textsf{reference}$! We use \textsc{Index} Algorithm to reallocate all edges in machines (each edge $\{u, v\}$ will be considered as $(u,v)$ and $(v,u)$). For example, $N(v_1)$ will be stored in consecutive machines with indexes in $[1, d(v_1)]$. From the $\textsf{reference}$, we know the index of the node $v$ in $S$ (if $v\in S$).

After finishing the above procedure, we now show how to make each node in each machine know its index in $S$. We let each node $u$ in each machine access its index $S$ by querying $\textsf{reference}$ for  node $u$. There is no congestion because different queries will go to different machines or spaces. Then we achieve our goal. 
\end{proof}

\noindent \textbf{\textsc{ReorganizeNBR}}~~The ReorganizeNBR problem involves an input set $S$, where the objective is to reorganize $N(u)$ for all $u\in S$ in a manner that \emph{aligns} the bytes of $N(u)$ and $N(v)$ for any two nodes $u, v\in S$. Specifically, the $i$-th byte of $N(u)$ indicates whether the $i$-th node is a neighbor of $u$. The motivation behind the ReorganizeNBR problem is to efficiently compute $|N(u)\cap N(v)|$ for any pair of nodes $u$ and $v$ in $S$ (refer to Figure~\ref{fig:mpc-align} for illustration). We utilize the \textsc{ReorganizeNBR} algorithm to address this problem.

\begin{figure}[ht]
  \includegraphics[width = 155mm]{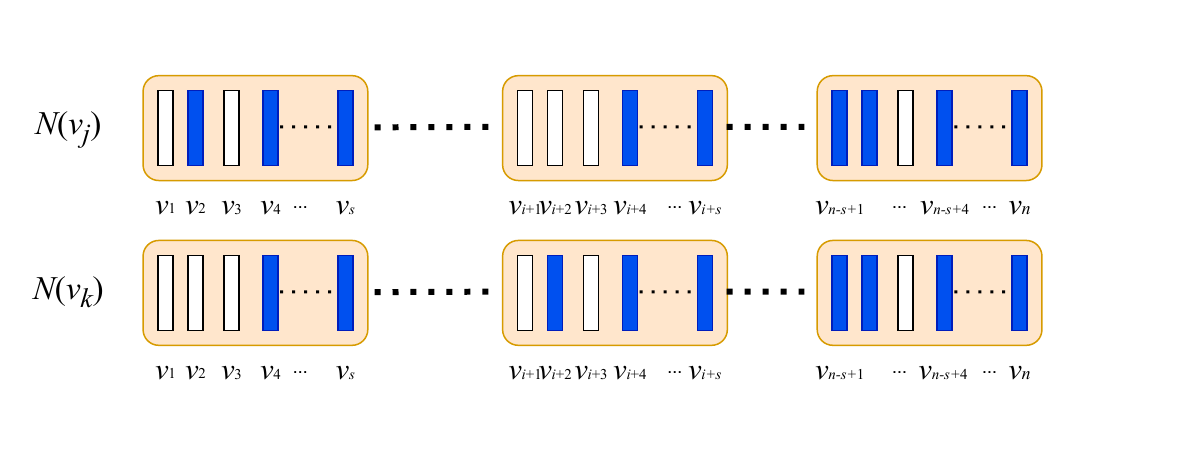}
\caption{Align Operation. The first row represents $N(v_j)$, the blue rectangle (treated as ``$1$'') indicates that the corresponding vertex $v_i \in N(v_j)$, and the white rectangle (treated as ``$0$'') 
    indicates that the vertex $v_i \not \in N(v_j)$. Once 
  $N(v_j)$ and $N(v_k)$ are encoded as such bit strings, we can compare the strings simultaneously to compute their common neighbors. }\label{fig:mpc-align}
\end{figure}

As a warm-up, we first show how to reorganize $N(u)$ for each node $u$ in the graph $G$ with $m = \Theta(n^2)$ where $m$ is the number of edges and $n$ is the number of nodes in  $G$.

\begin{restatable}{lemma}{orga}\label{lem: orga}
Given a graph $G = (V, E)$ with $m= \Theta(n^2)$, there exists an algorithm in MPC model where the memory of each machine is $s = \Omega(\log n)$ that can reorganize $G$ in constant rounds, where $m=|E|$ and $n=|V|$. 
\end{restatable}

\begin{proof}
We give each unit space of $N/s$ machines an index, where $s$ is the size of memory of each machine and $N = O(m+n)$. For example, in the $i$-th machine, the indices are from $(i-1)s+1$ to $is$. Although the input graph is undirected, we consider it as a directed graph, i.e., each edge has double directions. Let $v_1,\ldots,v_n$ be those nodes in $V$. When we organize these nodes, the priorities are $v_1 > v_2 > \cdots > v_n$. We describe how arranging $n/s$ machines to store $N(v)$ for a node $v$ with priorities works. Notice that this procedure is not sorting. Instead, we reorganize the data. 
We use $n$ bits to store each $N(u)$ for each $u\in V$. 
 For each edge $(v_i, v_j)$, there is an index $ni+j$ that can be viewed as its increasing order. For such an edge $(v_i,v_j)$, we put it into the $(ni+j)$-th unit of the memory. Since each machine has $s$ words and the labels of machines are $1,\ldots,p$, the edge $(v_i,v_j)$ will be stored in the $\lfloor \frac{ni+j}{s} \rfloor$-th machine. Similarly, we execute the procedure for $(v_j,v_i)$ (each edge has two directions). The whole process only takes constant rounds. 

\end{proof}

Now, we show how to reorganize a graph in MPC model where each machine has memory $s =\Omega(\log n)$. 
We say a subgraph $H$ is a randomly sampled subgraph if nodes in $V_H$ are randomly sampled, and $H$ is the set of edges and vertices constructed from picking $V_H$ with their incident edges. We use $H=(V_H, E_H)$ to denote such a random sampled graph. 

\begin{restatable}{lemma}{orgmain}\label{lem: org-main}
Given a graph $G = (V, E)$ with $ m= \Theta(n^{c+1})$, where $m$ is the number of edges and $n$ is the number of vertices, and $c \in (0,1]$ is some positive constant, for a given randomly sampled subgraph $H = (V_H, E_H)$ satisfying $|V_H| \leq \tilde O(m/n)$ where $V_H$ is constructed by \textsc{RandomSet}, there exists an algorithm that can reorganize $N(u)$ for each $u \in V_H$ in constant rounds, with total space complexity $\tilde O(m)$ in the MPC model where each machine has a memory of $s=\Omega(\log n)$.  
\end{restatable}

\begin{proof}

In our setting, we use \textsc{RandomSet} to create a set $S$, i.e., $V_H$ of random numbers in $[n]$. By Lemma~\ref{lem:randomset}, each machine $p_i$ can know the indexes of nodes in $p_i \cap V(H)$ in $V_H$ within $O(\log_s n)$ and thus can execute the same procedure as follows. Then we map each node in $S \cap p_i$ to the corresponding space. For example, the node $u \in S$ with $I_j$ has a neighbor $v \in S$ with $I_k$. We map $v$ to the $(nI_j + I_k)$-th space and map $u$ to the $(nI_k + I_j)$-th space. We repeat this procedure until all nodes are mapped. Therefore, we get consecutive machines to store $N(u)$ for all $ u\in V_H$ (the details can seen in the proof of Lemma~\ref{lem: orga}). Notice that, we keep empty positions which are used to make the positions of neighbors of each node $u$ align, such that we can compare $N(u)$ to $N(v)$ for any two nodes $u$ and $v$ efficiently. Now, let us consider the space we used. $|N(u)| = n$ for each $u\in V_H$. As there are $|S|$ nodes, the total space we used is $O(n|S|)$, i.e., $O(n|V_H|)$. Since $|V_H| \leq \tilde O(m/n)$, $n|S|\leq \tilde O(m)$.  The whole process only takes constant rounds. 
\end{proof}

\noindent \textbf{\textsc{CopyNBR}($\mathcal{S},t$)}~~Suppose we have $\mathcal{S} = \{N(v_1),\ldots, N(v_x)\}$ stored in machines where $v_i \in S$, $i \in[x]$ and $S$ is a random set created by $\textsc{RandomSet}$. We will create $N(v_i)_1,\ldots, N(v_i)_t $ for each $N(v_i)$ where $i\in[x]$.
The goal is to make $t$ copies of $N(u)$  for each $u\in S$ such that we can execute other algorithms in parallel. In our setup, each $N(u)$ where $u\in S$ is organized in a collection of consecutive machines. To solve this problem, we employ the \textsc{CopyNBR}($\mathcal{S},t$) algorithm. Upon executing \textsc{CopyNBR}($\mathcal{S},t$), all $t$ copies of $\mathcal{S}$ are stored in consecutive machines. 

\begin{restatable}{lemma}{copymain}\label{lem: copy-main}
The CopyNBR($\mathcal{S},t$) problem can be solved in $O(\log_s n)$ rounds in the $s$-space MPC model where $s = \Omega(\log n)$ and $t$ is a parameter satisying $n|S|t \leq \tilde O(m)$.
\end{restatable}
\begin{proof}
By Theorem \ref{them:basics}, we can make $t$ copies of $\mathcal{S}$ for $t$ times within $O(\log_s n)$ rounds. Then we  reorganize these copied sets by making each copy of $\mathcal{S}$ stored in consecutive machines, which can be finished in at most $O(\log_ s n)$ rounds.
\end{proof}

\noindent \textbf{{EvenCluster}}~~Consider a set $S$ comprising nodes labeled from $1$ to $k$. The objective is to ensure that the number of nodes with labels in $S$ is even. To achieve this, we employ the \textsc{EvenCluster} algorithm, designed specifically to solve this problem.

\begin{restatable}{lemma}{evencluster}\label{lem: evencluster}
In the MPC model with each machine's memory $s = \Omega(\log n)$, there exists an algorithm that can output $S' \subseteq S$ within $O(\log_s n)$ rounds, such that each label in $S'$ is associated with the same number of nodes with that label. 
\end{restatable}
\begin{proof}
    Let $\xi$ be the number of nodes with the same label in $S'$. Our algorithm proceeds in three slots.  Let machines with indexes $[\frac{ni}{s}, \frac{n(i+1)-1}{s}]$ deal with the messages in which the label is $i$.  In the first slot, we only let head machines with respect to $N(u)$ send their labels to the $\lfloor \textsf{ID}(u)/s \rfloor$-th machine. In the second slot, we use \textsc{Sorting} algorithm to sort all nodes in $S$ with the same label. From \cite{goodrich2011sorting}, \textsc{Sorting}  in MPC model can be finished in~$O(\log_s n)$ rounds. In the last slot, we send the sorting results back to head machines. If the index of the sorting result is larger than $\xi$, the head machine along with other machines storing $N(u)$ where $u\in S$ become inactive. 
\end{proof}

\noindent \textbf{\textsc{RepresentativeK}($S$)}~~In the {RepresentativeK}($S$) problem, the input is a set $S={S_1\cup \cdots \cup S_k}$  of nodes with $|S|$ labels (each node in $S_i$ has $|S_i|$ labels) where $S_i\cap S_j = \emptyset$ for any $i\not = j$ and $|S_i| = \Theta(|S|/k) \geq \Omega(\log n)$.  Our goal is to output $k$ nodes with $k$ representative labels. We use \textsc{RepresentativeK}($S$) to solve this problem.

\begin{restatable}{lemma}{unik}\label{lem:unik}
    The {RepresentativeK}$(S)$ problem can be solved in $O(\log_s n)$ rounds where $S$ is created by $\textsc{RandomSet}$ and $|S| \geq \Omega(k\log n)$ in the $s$-space MPC model $(s=\Omega(\log n))$.
\end{restatable}

\begin{proof}
For each label, with high probability, there are $\Omega(\log n)$ corresponding nodes by Chernoff Bound. We find the minimum label of each $u$ and then we keep that label active. Otherwise, we make labels inactive. Then, each node will have one unique label and each node in the same cluster will have the same label. By \textsc{Sorting} Algorithm in MPC model, we can sort these $|S|$ nodes within $O(\log_s n)$ rounds.  Then, we can select $k$ nodes with unique labels. The round complexity is $O(\log_s n)$.
\end{proof}

\noindent \textbf{CompareCut($S,V$)}~~In the CompareCut($S, V$) problem, the input is a set $S$ of $k$ sets, i.e., $\{S_1,S_2,\ldots,S_k\}$ and the vertex set $V$, the goal is to output the largest one among numbers $n(u,S_i)$ of edges between $u\in V$ and $S_i$ for each $S_i$, along with the label of $u$ (i.e., the label of the $S_i$), where $i\in[k]$.
We use $\textsc{ComputeCut(S,V)}$ to solve this problem.

\begin{restatable}{lemma}{cmpcut}\label{lem:cmpcut}
    Given a graph $G = (V, E)$, let $S \subset V$ be a random set of nodes created by $\textsc{RandomSet}$. The CompareCut$(S, V)$ problem can be solved in $O(\log_s n)$ rounds in the $s$-space MPC model $(s = \Omega(\log n ))$.
\end{restatable}

\begin{proof}
By Lemma~\ref{lem:randomset}, each machine $p_i$ knows $p_i \cap S$. We let each machine reorganize~$N(u)$ for each $u\in S$ within constant rounds by Lemma \ref{lem: org-main}. We reorganize $N(u)$ $(u\in S)$ in consecutive machines according to order of $k$ clusters within $O(\log_s n)$ rounds (using \textsc{Sorting} Algorithm, we can index each label among $k$ labels within $O(\log_s n)$ rounds). Our goal is to find the number of edges between all $S_i$ and $V$ where $i\in[k]$. Now, each $N(u)$ for $u\in S$ is stored in consecutive machines. Next, we introduce the method to obtain the number of edges, i.e., $v_i \in V$ connecting to $S_i$ where $i\in[k]$. We count the number of $v_i$ in the $i$-th space of each $N(u)$ where $u\in S_i$. In practice, in the machines storing $N(u)$ where $u\in S_i$, we count the number of the nonempty $i$-th space for all corresponding machines. Then we sum the number of the non-empty $i$-th space up by the broadcast method. For example, a machine $p_1$ stores $v_1,\ldots,v_s$. We let the machine $p_1$ send $v_i$ to $p_{1+\frac{n(i-1)}{s}}$ where $i\in[s]$. Then $p_{1+\frac{n(i-1)}{s}}$ will receive at most $s$ messages containing $v_i$. Next, $p_{1+\frac{n(i-1)}{s}}$ calculates the sum, which is the number of $v_i$ and then sends this sum back to previous senders ($p_1,\ldots,p_{1+n\cdot (s-1)/s}$). Thus, each machine $p_{1+\frac{n(i-1)}{s}}$ ($i\in[s]$) will know the number of $v_1,\ldots,v_s$ for nodes among $s$ corresponding machines after two rounds. After repeating this kind of procedures for $O(\log_s n)$ rounds, we get the number of edges between $v_1$ to $S_i$. Since we run algorithms in parallel, within $O(\log_s n)$ rounds, we obtain $v_i$ to $S_j$ where $i\in[n]$ and $j\in[k]$. Therefore, we can know the number of edges between $V$ and $S_i$ for $i\in[k]$ within $O(\log_s n)$ rounds.  Then, we compute the largest one, which can be done by broadcast in $O(\log_s k)$ rounds (For every $s$ values, we get a maximal value. After $O(\log_s k)$ rounds, we obtain the largest value). Then, we get the label for $u$ from some $S_i$ with the largest value. 
\end{proof}

\section{The Algorithm Based on Neighbor Counting}\label{sec:mfirst-mpc-alg}
Recall that a graph $G=(V,E)$ is generated from the SBM$(n,p,q,k)$ if there is a hidden partition~$V=\cup_{i=1}^kV_k$ of the $nk$-vertex set $V$, and for any two vertices $u,v$ that belong to the same cluster, the edge $(u,v)$ appears in $E$ with probability $p$; for any two vertices $u,v$ that belong to two different clusters, the edge $(u,v)$ appears in $E$ with probability $q$, where $0< q<p< 1$. In this section, we give the algorithm underlying Theorem \ref{thm:firstmpc}. 

We first give a simple sequential algorithm based on comparing common neighbors. Then we show how to implement it in the $s$-space model. To do so, we give implementations of a number of basic graph operations in the $s$-space model, which is deferred to Section \ref{sec:mbasicmpc}.

\subsection{A Sequential Algorithm Based on Counting Common Neighbors}

\begin{restatable}{theorem}{firstmain} \label{thm:firstmain}
If $\frac{p-q}{\sqrt{p}} \geq \Omega({\frac{(k+1)^{1/2}}{n^{1/4}}})$, 
the algorithm \textsc{CommNBR} can output $k$ clusters in  $O(\frac{k^2n\log n}{p})$ time with probability $1-O(\frac{1}{n})$.
\end{restatable}

In our sequential algorithm \textsc{CommNBR}, we first randomly sample a set $S$ of~$\frac{21n k^2\log n}{d}$  nodes from $V$ such that each cluster has more than $\Theta({\log n})$ nodes with high probability where $d$ is the number of neighbors of an arbitrary node $u\in V$. Then, for each pair $u,v$, we count the number of their common neighbors in $G$, i.e., those vertices that are connected to both $u$ and $v$. If the number of common neighbors is above some threshold $\Delta$, then we put them into the same cluster. In this way, we can obtain $k$ \emph{sub-clusters} of $S$, $C_1,\dots,C_k$. That is, each $C_i\subseteq S$ and is a subset of some cluster, i.e., $C_i=V_{\pi(i)}\cap S$ for some permutation $\pi:\{1,\dots,k\}\to\{1,\dots,k\}$. 
Let $L(v,C_i)$ denote the number of incident edges between a node $v$ and a cluster $C_i$. We can then cluster each remaining node $v\in V\setminus S$ by finding the index $j$ such that $L(v,C_j)$ is the greatest among all numbers $L(v,C_i), 1\leq i\leq k$.

For the intuition of the existence of such a  threshold $\Delta$, let us take the case $k=2$ as an example. In this case, for any two vertices $u,v$ belonging to the same cluster, the expected number of common neighbors is $p^2n+q^2n$; for any two vertices $u,v$ belonging to two different clusters, the expected number of common neighbors is $2npq$. Since $\frac{p-q}{\sqrt{p}} \geq \Omega\left(\frac{(k+1)^{1/2}}{n^{1/4}}\right)$, there exists a sufficiently large gap between these two numbers so that we can define a suitable threshold. However, the values of $p$ and $q$ are not provided. 
To address this issue, we propose an algorithm called $\textsc{ComputeDEL}$, which can be described as follows. 
We first sample a set $\mathcal{S}_{\Delta}$ of $\Theta(k\log n)$ nodes, and for each pair  $u,v \in \mathcal{S}_{\Delta}$, we compute a set $\mathcal{V}_{\Delta}$ of values of $|N(u)\cap N(v)|$. We let $\Delta' = \max \{ value \in \mathcal{V}_{\Delta}\} $, and then we have $\Delta = \Delta' - 9\sqrt{\Delta'\log n }$.


\begin{algorithm}[h]
\caption{\textsc{CommNBR}($G,n,p,q,k$): A sequential algorithm based on counting common neighbors}
\label{alg:optsimple}
\KwInput{ A SBM graph $G$;}
\begin{algorithmic}[1]
   
   \STATE{Select an arbitrary node $u\in V$ and set $d = |N(u)|$}
   \STATE{Sample a set $S$ of $21n k^2\log n/d$  nodes in $V$}
 
  \STATE{Obtain $\Delta$ by \textsc{ComputeDEL}($G$)}
   \STATE Find $k$ clusters from $S$ by \textsc{CompcomNBR}($S,G,\Delta$)
   \STATE Let $C_1, C_2, \cdots, C_k$ be the obtained $k$ sub-clusters

     \STATE For each $i\leq k$, choose an arbitrary subset $\overline{C}_i\subseteq C_i$ such that $|\overline{C}_i|=20n k^2\log n/d$ 
   \FOR{$v \in V\setminus S$}
   \STATE Put $v$ in $C_i$ such that $i=\arg\max_j L(v,\overline{C}_j)$, where $L(v,\overline{C}_j)$ is \\ the number of edges between $v$ and $\overline{C}_j$ 
   \ENDFOR   
   \STATE Output $k$ clusters.
\end{algorithmic}
\end{algorithm}

In Algorithm $\textsc{CompcomNBR}$, while $S$ is not empty, we execute the following procedure. First, we choose an arbitrary vertex $v$ in $S$ and put $v$ into a set $C_i$. For each node $u$ in $S$, if $|N_G(u) \cap N_G(v)| \geq \Delta$, we add $u$ to $C_i$. Then, we remove~$C_i$ from the set $S$. Finally, we return all the sub-clusters $C_i$ ($i\in[k]$).

\subsubsection{Analysis of \textsc{CommNBR}}

We first give the estimation of $\frac{21k^2n\log n}{d}$.

\begin{lemma}
Let $u \in V$ be a vertex of the random graph $SBM(n,p,q,k)$. We have that, for any~$\alpha\in (0,\sqrt{(n-1)p+n(k-1)q})$, 
$
    \mathrm{P_r} (|np+nq(k-1)-d(u)| \geq \alpha \sqrt{np+nq(k-1)]} ) \leq e^{-\alpha^2/3}.
$
\end{lemma}

\begin{proof}
For any vertex $u\in V$, the expectation of the degree of $u$, i.e., $d(u)$ is $(n-1)p+n(k-1)q$. We can see that $d(u)$ is the sum of $kn-1$ independent variables, $x_1,x_2,\ldots, x_{kn-1}$ where $x_i$ indicates whether the $i$-th edge from $u$ exists or not. By Chernoff bound, for any $\alpha\in (0,\sqrt{(n-1)p+n(k-1)q})$, we have
$$
    \mathrm{P_r} (|(n-1)p+nq(k-1)-d(u)| \geq \alpha \sqrt{(n-1)p+nq(k-1)]} ) \leq e^{-\alpha^2/3}.$$
Thus,
$$
    \mathrm{P_r} (|np+nq(k-1)-d(u)| \geq \alpha \sqrt{np+nq(k-1)]} ) \leq e^{-\alpha^2/3}.
$$
\end{proof}

Therefore, $d(u) \in [np+nq(k-1)-\alpha \sqrt{np+nq(k-1)}, np+nq(k-1)+\alpha\sqrt{np+nq(k-1)}]$ with probability at least $1-e^{-\alpha^2/3}$. 
As $p <  p+(k-1)q < kp$, we have Corollary~\ref{cor:deg} by setting $\alpha = \sqrt{3\log n}$. 
\begin{corollary}\label{cor:deg}
For any vertex $u$ in a random graph generated by SBM($n,p,q,k$), with probability at least $1-\frac{1}{n^3}$, $\frac{21k^2n\log n}{d(u)} \in  ( \frac{20k\log n}{p},  \frac{22k^2\log n}{p})$.
\end{corollary}

Next, we show the following lemma, which says counting the number of common neighbors is a good strategy to decide whether two vertices belong to the same cluster.  Let $C(v_i)$ denote the cluster that contains $v_i$ where $i\in[n]$.

\begin{lemma}\label{lem:com-neigh}
Given a graph $G$ that is generated from SBM$(n, p,q,k)$ such that $\frac{p-q}{\sqrt{p}} \geq \frac{6(k+1)^{1/2}(\log n)^{1/4}}{n^{1/4}}$ and $k$ is the number of hidden clusters. With probability at least $1-\frac{1}{n^2}$, for any two nodes $u$ and $v$, if $C(u) = C(v)$, then $|N(u)\cap N(v)| > \Delta $ ; Otherwise, $|N(u)\cap N(v)| < \Delta $. 
\end{lemma}

\begin{proof}
Consider three nodes $u$, $v$ and $w$ in $V$. Without loss of generality, we assume that $u$ and~$v$ are in the same cluster, and $w$ is in another cluster. Then, we have the following equalities. 
$$\mathbb{E}[|N(u)\cap N(v)|] = p^2(n-2)+q^2n(k-1)$$
$$\mathbb{E}[|N(u)\cap N(w)|] = 2(n-1)pq + q^2n(k-2) $$
$$\mathbb{E}[|N(u)\cap N(v)| - |N(u)\cap N(w)|] = n(p-q)^2 - 2p(p-q)$$

For convenience, let $s_1 = p^2(n-2)+q^2n(k-1)$ and $s_2 = (n-1)pq + q^2n(k-2)$.


Next, let us see the sufficient condition.
By Chernoff Bound, with probability at least $1-\frac{1}{n^3}$, $$|N(u)\cap N(v)| \geq t_1=  s_1 - 6\sqrt{s_1\log n }. $$ 
Similarly, we have  $$|N(u)\cap N(w)| \leq t_2 = s_2 + 16\sqrt{s_2\log n }$$ with probability at least $1-\frac{1}{n^3}$.

Recall that by \textsc{ComputeDEL}, we set $\Delta$ to be the maximum value among $\mathcal{V}_{\Delta}$. We claim that~$\Delta \in (s_1-16\sqrt{s_1\log n}, s_1-6\sqrt{s_1\log n})$ with probability at least $1-\frac{1}{n^3}$. From the process in \textsc{ComputeDEL}, it is easy to see that with probability at least $1-\frac{1}{n^3}$,  $\Delta' = |N(u)\cap N(v)|$ for some $u,v \in \mathcal{S}_{\Delta}$ and $u, v$ are in the same cluster. By Chernoff Bound, we have 
$$\Delta' \in (s_1 - 3\sqrt{s_1\log n} , s_1 + 3\sqrt{s_1 \log n})$$
with probability at least $1-\frac{1}{n^3}$. Thus, we prove our claim by $\Delta = \Delta' - 9\sqrt{\Delta'\log n }$. We can see that $t_1 > \Delta$.

Next, we will show that $\Delta >  t_2$, given $\frac{p-q}{\sqrt{p}} \geq 6\frac{(k+1)^{1/2}(\log n)^{1/4}}{n^{1/4}}$. We have
$$\Delta - t_2 > (\sqrt{s_1}-\sqrt{s_2}-16\sqrt{\log n})(\sqrt{s_1}+\sqrt{s_2})$$

Since $\frac{p-q}{\sqrt{p}} \geq 6\frac{(k+1)^{1/2}(\log n)^{1/4}}{n^{1/4}}$, 
$$
 s_1-s_2=n(p-q)^2-2p(p-q)\geq32(k+1)p\sqrt{n\log n}+9\log n
$$

Then, we get that $s_1-s_2 > 32\sqrt{s_1\log n}+9\log n$, so 
$$
\sqrt{s_1}-\sqrt{s_2}-16\sqrt{\log n} > 0
$$

So by union bound, with probability at least $1-\frac{1}{n^2}$, $|N(u)\cap N(v)| > \Delta > |N(u)\cap N(w)|$ holds for any $u, v, w$ where $C(u) = C(v)$ and $C(u) \not = C(w)$. It means that we can cluster different clusters by the number of common neighbors and the threshold is $\Delta$. 
\end{proof}

Therefore, $|N(u)\cap N(v)| > |N(u)\cap N(w)|$ holds with high probability for any $u, v, w$ where $C(u) = C(v)$ and $C(u) \not = C(w)$. It means that we can cluster different clusters by the number of common neighbors and the threshold is $\Delta$. Recall that $L(v,C_i)$ is the number of edges between a node $v$ and a cluster $C_i$. Now we show the following lemma.

\begin{lemma}\label{lem:remaincluster}
Under the same condition as the one in Lemma \ref{lem:com-neigh}, for any vertex $v\in \overline{C}_i$, with probability at least $1-O(\frac{1}{n^2})$, it holds that $L(v,\overline{C}_i) > L(v,\overline{C}_j)$ for any $i$ and $j$ ($i\not =j$).
\end{lemma}

\begin{proof}
Consider any two clusters $\overline{C}_i$ and $\overline{C}_j$. Let $Y_i$ be the number of nodes in $\overline{C}_i$. We have that $Y_i = Y_j= Y = \frac{20k^2n \log n}{d}  >  \frac{9\log n}{p}$ for any $i,j\in [k]$ with probability at least $1-\frac{1}{n^3}$. If $i\not = j$, by Chernoff bound, with probability $1-O(1/n^2)$ we get that  
\begin{align*}
&L(v,\overline{C}_i) - L(v,\overline{C}_j)\\
> &pY_i(1-\frac{\sqrt{9\log n}}{\sqrt{pY_i}}) - qY_j(1+\frac{\sqrt{9\log n}}{\sqrt{qY_j}})\\
= &Y(p - \sqrt{\frac{9p\log n}{Y}} - (q - \sqrt{\frac{9q\log n}{Y}}))\\
= &Y(\sqrt{p}-\sqrt{q})(\sqrt{p}+\sqrt{q}-\sqrt{\frac{9\log n}{Y}})\\
>& Y(\sqrt{p}-\sqrt{q})\sqrt{q}
\end{align*}

Then, $L(v,\overline{C}_i) - L(v,\overline{C}_j) > 0$ with probability $1-O(1/n^2)$.
\end{proof}

\begin{proof}[\textbf{Proof of Theorem \ref{thm:firstmain}}]
The number of sampling nodes is $|S| \geq \frac{10 k\log n}{p}>\Omega(k\log n)$, so each cluster will have at least $\frac{9\log n}{p}$ nodes with probability $1-O(1/n^3)$. By Lemma \ref{lem:com-neigh}, we can cluster nodes via comparing their common neighbors. In Algorithm~\ref{alg:optsimple}, we obtain $k$ clusters by comparing common neighbors sequentially which takes at most $n|S|$ time. Therefore, this step takes $O(k^2n\log n/p)$ time.

     By Lemma~\ref{lem:remaincluster}, we can cluster the remaining nodes correctly with probability $1-O(1/n)$ by taking the union bound of at most $n$ nodes. For each node $v\in V\setminus S$, we count the number of connected edges between $v$ and each cluster $\overline{C}_i$ where $i\in[k]$. For each such operation, we need $|S|$ time. Then, the total time of clustering the remaining nodes is $(n-|S|)|S|\leq O(\frac{k^2n\log n}{p})$. Putting all together, we prove that Algorithm~\ref{alg:optsimple} outputs $k$ clusters with probability $1-O(1/n)$ and the time complexity of Algorithm~\ref{alg:optsimple} is $O(\frac{k^2n\log n}{p})$.
\end{proof}

\subsection{Implementation in the $s$-space MPC model}\label{subsec:m-mpc-1}

Now we describe our MPC algorithm \textsc{MPC-CommNBR}, which is an implementation of \textsc{CommNBR},  
where the local memory is $s = \Omega(\log n)$, and prove Theorem \ref{thm:firstmpc}. 

Recall that in  \textsc{CommNBR}, there are two major steps. In the first step, we need to find $k$ clusters from a set $S$ of randomly sampled nodes. In the second step, based on the clustering on $S$, we cluster all nodes in~$V$. The major challenge lies in the simulation (in MPC model) of the first step which is to compare common neighbors between two nodes $u$ and~$v$. It is easy to see that computing~$|N(u)\cap N(v)|$ for $u,v\in V$ is exactly the task of finding common elements in two sets. For convenience, we use a set $\mathsf{S}_u$ to denote~$N(u)$ for a node $u\in V$. Then, we need to find the common elements between $\mathsf{S}_u$ and $\mathsf{S}_v$ by a method $\textsc{comm}(\mathsf{S}_u, \mathsf{S}_v)$. Note that we need to execute $\textsc{comm}(\mathsf{S}_u, \mathsf{S}_v)$ for different $u$ and $v$ for many times. 
Therefore, to compute $|N(u)\cap N(v)|$ for different $u,v \in V$ efficiently, we need to solve two problems. The first one is to implement $\textsc{comm}(\mathsf{S}_u, \mathsf{S}_v)$ efficiently in MPC model where each machine's memory is ~$s=\Omega(\log n)$. The second one is to execute the first one in parallel.  For the first one, we use a simple method to implement $\textsc{comm}(\mathsf{S}_u, \mathsf{S}_v)$ in MPC model. Let~$V'$ be the set of nodes in which for each $u\in V'$, $\textsf{S}_u$ will be compared. We make each byte of~$\textsf{S}_u(u\in V')$ aligned. Then, we can directly compute $|\textsf{S}_u \cap \textsf{S}_v|$. For the second one, we solve it by copying sets for enough times and then we let machines storing these copied sets execute the same algorithm in parallel. We use the MPC implementations of the basic graph operations in Section \ref{sec:mbasicmpc} to implement our clustering algorithm here.

\begin{algorithm}
\caption{\textsc{MPC-CommNBR}: An MPC algorithm based on counting common neighbors}
   \label{alg:mpc-1}
   \KwInput{ A SBM graph $G$;}
\begin{algorithmic}[1]
   \STATE{Let $d = |N(u)|$ where $u$ is an arbitrary node in $V$}
   \STATE {Apply \textsc{RandomSet} to obtain random node set $S'$ and $S$, where $|S'| = \Theta(k\log n)$ and $|S| = (21k^2n\log n)/{d}$}
   \STATE{Update $S'$  and obtain $\Delta$ by \textsc{ComputeREP}($S'$)}
     \STATE{Obtain $k$ sub-clusters  $S_1,\ldots,S_k$ by \textsc{ComputeSubcluster}($S,S', \Delta$)}
   \STATE{Obtain $k$ clusters of $V$ by \textsc{ComputeCluster($S_1,\ldots,S_k, V$)}}.
\end{algorithmic}
\end{algorithm}

\begin{algorithm}[h]\caption{\textsc{ComputeRep}: Compute representative for each cluster by common neighbors and $\Delta$}
\KwInput{Random vertex set $S'$;}\label{compR}
\begin{algorithmic}[1]

    \STATE Run \textsc{CompareINIT}($S'$);
    \STATE Obtain $k$ nodes with $k$ representative lables by $\textsc{RepresentativeK}(S')$\;
    \STATE Update $S'$ by only keeping $k$ nodes obtained from Step 2.
    \RETURN{$S'$ and $\Delta$}
\end{algorithmic}
\end{algorithm}

For the algorithm \textsc{CompareINIT}($S'$), we describe it as follows.
\begin{center}
\fbox{
\begin{minipage}{\linewidth-2em}

\noindent\textsc{CompareINIT($S'$)}:
\begin{itemize}
\item [(I)] Reorganize neighbors of nodes in $S'$ by \textsc{ReorganizeNBR($S'$)}.
\item [(II)] Execute \textsc{CopyNBR($S'$)} to create~$|S'|$ copies of~$N(u)$ for each $u\in S'$.
\item [(III)] Based on copies of~$N(u)$ from (II), we directly compute ~$|N(u)\cap N(v)|$ in parallel where $u,v \in S'$. 
\item [(IV)] Execute \textsc{CompareGRP} to obtain final results by summing all partial results obtained from (III). 
\item [(V)] Compute $\Delta'$, i.e., the maximum  value of  $|N(u)\cap N(v)|$ for all pairs of $u,v \in S'$ and output $\Delta = \Delta' - 9\sqrt{\Delta'\log n}$.
\end{itemize}
\end{minipage}

}
\end{center}

In \textsc{ComputeSubcluster}($S',S, \Delta$), the process is similar to  \textsc{ComputeREP}($S'$). The major difference is that we only copy $N(u)$ for each $u \in S$ for 
$k$ times and copy $N(v)$ for each $v\in S'$ for $|S|$ times. By computing~$|N(u)\cap N(v)|$ where $u\in \textsc{copy-}S$, $v\in \textsc{copy-}S'$, and $\textsc{copy-}S$, $\textsc{copy-}S'$ are  the copies of $S$ and $S'$, we can obtain $k$ sub-clusters of $S$. The details of computing can refer to \textsc{ComputeREP}($S'$). 

In \textsc{ComputeCluster}($S_1,\ldots,S_k, V$), we first apply \textsc{EvenCluster}($S$) to output $k$ sub-clusters from $S$ such that each cluster has the same number of nodes. Then, we use \textsc{ComputeCut}($S, V$) to cluster $V$.

Next, we show the details of $\textsc{CompareGRP}$ used in the  procedure of \textsc{ComputeREP}($S'$). In the~$\textsc{CompareGRP}$ problem, the input is a set of groups of machines and the goal is to output the results by comparing groups of machines. Take two groups $A$, $B$ of machines as an example. Our goal is to output the common elements by comparing $A$ and $B$. We say that group $A$ compares to group $B$ which means that the $i$-th member machine of the group $A$ will compare to the $i$-th member machine of group $B$ ($i\in[n/s]$).

\begin{restatable}{lemma}{compmain}\label{lem: comp-main}
Given a set of consecutive groups each of which has $n/s$ member machines, there exists an algorithm that takes $O(\log_s n)$ rounds to obtain the results of comparing data between groups correspondingly. 
\end{restatable}
\begin{proof}
 To compare data between groups, we let each machine send the whole data to the target machine. Then the target machine receives the data and has a partial result. The next step is to accumulate all these partial results. Using the converge-cast method, we finish this step within $O(\log_s n)$ rounds. 
\end{proof}

Now we are ready to prove Theorem~\ref{thm:firstmpc}.
\begin{proof}[\textbf{Proof of Theorem~\ref{thm:firstmpc}}]
The correctness of obtaining $k$ clusters based on counting common neighbors can be seen in Theorem~\ref{thm:firstmain}. By Lemma~\ref{lem:randomset}, we can create randomly sampled sets $S$ and $S'$ such that each machine $M$ knows indexes of nodes in $M\cap S$ and $M\cap S'$ within $O(\log_s n)$ rounds. 

Next, we first prove that by $\textsc{ComputeREP}(S')$, we can obtain $k$ sub-clusters within $O(\log_s n)$ rounds. By Lemma~\ref{lem: org-main}, it takes $O(\log_s n)$ rounds for $\textsc{ReorganizeNBR}(S')$. By Lemma~\ref{lem: copy-main}, we use~$O(\log_s n)$ rounds to finish $\textsc{CopyNBR}(S')$ for each $N(u)$ where $u\in S'$. By Lemma~\ref{lem: comp-main}, it takes~$O(\log_s n)$ to first get partial results and then obtain the complete results of~$|N(u)\cap N(v)|$ where~$u, v \in S'$. We can use $O(\log_s n)$ rounds to obtain $\Delta$ by simulating $\textsc{ComputeDEL}$ within~$O(\log_s n)$ rounds. Then by Lemma~\ref{lem: comp-main} and Lemma~\ref{lem:unik}, we can obtain $k$ sub-clusters from $S'$ in~$O(\log_s n)$ rounds in the MPC model and the total space is $\tilde O(m)$. Similarly, by $\textsc{ComputeSubcluster}(S', S, \Delta)$, we can prove that within $O(\log_s n)$ rounds, we can obtain $k$ clusters from $S$ in the MPC model and the total space used is $\tilde O(m)$.

Now, let us see the last step of obtaining $k$ clusters of $V$. By Lemma~\ref{lem: evencluster} and setting $|S^*|= \frac{20k^2n\log n}{d}$, we can output $S^* \subseteq S$ such that for any two labels $i, j \in [k]$, we have $\textsf{N}_{S^*}(i) = \textsf{N}_{S^*}(j)$ in~$O(\log_s n)$ rounds, where~$\textsf{N}_{S^*}(i)$ is the number of nodes in ${S^*}$ with label~$i$.  Finally, by Lemma \ref{lem:cmpcut}, we can decide all labels of $V$ within $O(\log_s n)$ rounds with high probability. The total space used in MPC model is $\tilde O(m)$. Thus, our proof is completed. 
\end{proof}

\section{The Algorithm Based on Power Iterations}\label{sec:msecond-mpc-alg}

In this section, we give another MPC algorithm for a general SBM graph in the $s$-space model and prove Theorem \ref{thm:secondmpc}. 
Also, we first carefully design a sequential algorithm and then we implement it in the $s$-space MPC model. Our second sequential algorithm is based on power iterations which perform well in a recent algorithm in \cite{MZ22SBM}. The algorithm makes use of the adjacency matrix $A$ of the graph $G$, from which we define a matrix $B=A-q\cdot J$, where $J$ is the $n\times n$ all-$1$ matrix. Then it decides if two vertices $u,v$ are in the same cluster or not by checking the $\ell_2$-norm of the difference between $B_u^r$ and $B_v^r$, which are the rows corresponding to vertices $u,v$, respectively, in the matrix $B^r$ (the $r$-th power of matrix $B$).

We note that the algorithm in \cite{MZ22SBM} only considers the special case that $r=\log n$. Here, our sequential algorithm considers all possible $r\in \{1,\dots,O(\log n)\}$ and for each $r$ we choose a different threshold $\Delta$, which 
depends on the value of $p, q$ and $k$. To implement our sequential algorithm in the MPC model, we divide the process of matrix computation into different components each of which can be implemented efficiently in the $s$-space model.

\subsection{A Sequential Algorithm Based on Power Iterations}

In this section, for any matrix $M$, we use $M_i$ to denote the $i$-th row of $M$. Then, we let $\lVert M \rVert$ denote the $l_2$-norm of the matrix; $\lVert M_i \rVert$ denote the $l_2$-norm of the $i$-th row of $M$; $\lVert M \rVert_{\row}$ denote the maximum of the $l_2$-norm of the rows of $M$.

We now describe Algorithm \textsc{PowerIteration}. Let $A$ be an adjacent matrix of the input graph $G$ and $r$ where $r$ is a parameter. Let $\Delta = C\sqrt{k}\sqrt{p(1-q)}(\log{kn})^7(p-q)^{r-1}n^{r-1}$, where $C>0$ is some universal constant. We set $B = A - q\cdot J$ where $J = 1^{n\times n}$. Let $W = V$. We choose an arbitrary vertex $v$ in $W$ and put $v$ into a sub-cluster $C_i$ where $i\in[k]$. For each node $u$ in $W$, if $\lVert B_u^r - B_v^r\rVert \leq \Delta$, then we add $u$ to $C_i$. Next, we remove $C_i$ from $W$. Repeat the above process until $W$ is empty. Then we return all the clusters $C_i$'s. 

\begin{algorithm}[h]
   \caption{\textsc{PowerIteration}($G,n,p,q,k$): Detecting Clusters via $r$ Power Iterations}
   \label{alg:general}
   \KwInput{ A SBM graph $G$, $\Delta$}
\begin{algorithmic}[1]
   \STATE Let $A$ be an adjacent matrix of $G$
   \STATE Let $r$ be the parameter
   \STATE Let $\Delta = C\sqrt{k}\sqrt{p(1-q)}(\log{kn})^7(p-q)^{r-1}n^{r-1}$, where $C>0$ is some universal constant
   \STATE $B = A - q\cdot J$ where $J = 1^{n\times n}$
   \STATE{let $i=1$ and $W=V$}
\WHILE{$W$ is not empty}
\STATE{choose an arbitrary vertex $v\in W$}
\STATE{let $C_i=\{v\}$}
\FOR{each $u\in W$
}
   \IF{$\lVert B_u^r - B_v^r\rVert \leq \Delta$}
   \STATE{add $u$ to $C_i$}
   \ENDIF
   \ENDFOR
\STATE{$W\gets W\setminus C_i$}
\STATE{$i=i+1$}
\ENDWHILE
\STATE Return all the sub-clusters $C_i$'s.
\end{algorithmic}
\end{algorithm}

\begin{restatable}{theorem}{rpower}\label{them:r-power}
Let $p,q \leq 0.75$ be parameters such that $\max\{p(1-p),q(1-q)\} \geq C_0 \log n/n$ where $C_0>0$ is some constant. Suppose that $\frac{p-q}{\sqrt{p}} \geq (C_0^{2}+1) k^{\frac{1}{2}}n^{-\frac{1}{2}+\frac{1}{2(r-1)}}(\log kn)^7$. Let $G$ be a random graph generated from SBM($n,p,q,k$) and $r\in[3, O(\log n)]$, then with probability at least $1- O(n^{-1})$ the algorithm \textsc{PowerIteration} can output $k$ hidden clusters for suitable $\Delta$ and $r$.
\end{restatable}

The proof of Theorem~\ref{them:r-power} is built upon the work \cite{MZ22SBM}. We give a proof sketch here by making use of some of their subroutines. 
Recall that $A$ is the adjacency matrix of the graph $G$ obtained from $SBM(n,p,q,k)$ and $B = A - q\cdot 1^{n\times n}$. Following the notations in \cite{MZ22SBM}, we  decompose $B$ as $B = L + R$ such that $L = \mathbb{E}[B]$ is the``expectation part" of $B$, and $R$ is the ``random noise part". Note that the matrix $R$ is a symmetric matrix where each entry has mean $0$ and is independent of each other. 

Let $M = L^{-1}R + L^{r-1}RB + \cdots + LRB^{r-2}$ and $M' = RLB^{r-2}+R^2LB^{r-3}+ \cdots R^{r-1}L$. As noted in \cite{MZ22SBM}, it holds that 
\begin{equation}\label{eq1}
	B^r = (L+R)^r = L^r + M+ M' + R^r
\end{equation}

We use $i\sim j$ to denote that $i$ and $j$ are in the same cluster and use $i\nsim j$ to denote that they are not in the same cluster. In our setting, each cluster has the same size $n$. 
The following results were shown in \cite{MZ22SBM}.
\begin{lemma}[\cite{MZ22SBM}]\label{lem:Lij}
Let $v_i\in V_1$. For every $i\sim j$, $\lVert L^r_i - L^r_j \rVert=0$. Otherwise, if $i\nsim j$, then $\lVert L^r_i - L^r_j\rVert \geq 2\Delta$.
\end{lemma}

By setting $s^*=n$ in \cite{MZ22SBM}, we have the following lemma. 
\begin{lemma}[
\cite{MZ22SBM}]
\label{cor:maintool}
With probability at least $1- O(n^{-1})$, we have 
\begin{itemize}
    \item $\lVert M \rVert_{row} \leq r(192\sqrt{p(1-q)}\sqrt{kn}\log {kn}) \cdot \sqrt{n}\cdot (p-q)^{r-1}n^{n-2}$
    \item $\lVert M' \rVert_{row} \leq 2C_2\sqrt{p(1-q)}(\log{kn})^7 \sqrt{kn}\sqrt{n}(p-q)^{r-1}n^{r-2}$
    \item $\lVert R^r \rVert_{row} \leq (C_0 \sqrt{p(1-q)}\sqrt{kn})^r$
\end{itemize}
\end{lemma}

Now we are ready to prove Theorem~\ref{them:r-power}, which we re-state below.

\rpower*


\begin{proof}
From Equation (\ref{eq1}), we have
\begin{align*}
|\lVert B^r_i - B^r_j\rVert - \lVert L^r_i - L^r_j\rVert|
&\leq \lVert(B^r_i - L^r_i) - (B^r_j - L^r_j)\rVert
\\
&\leq \lVert M_i - M_j \rVert + \lVert M'_i - M'_j \rVert + \lVert R^r_i - R^r_j\rVert
\\
&\leq 2(\lVert M\rVert_{row} + \lVert M'\rVert_{row} + \lVert R^r\rVert_{row})
\end{align*}

By Lemma~\ref{cor:maintool}, with probability $1- O(n^{-1})$, we have the following
\[\lVert M \rVert_{row} \leq r\sqrt{kn}(192\sqrt{p(1-q)}\log n) \cdot \sqrt{n}\cdot (p-q)^{r-1}n^{n-2}\]
\[\lVert M' \rVert_{row} \leq 2C_2\sqrt{kn}\sqrt{p(1-q)}(\log{kn})^7 \sqrt{n}(p-q)^{r-1}n^{r-2} \]
\[\lVert R^r \rVert_{row} \leq (C_0 \sqrt{p(1-q)}\sqrt{kn})^r \]

Recall that $\Delta = C\sqrt{k}\sqrt{p(1-q)}(\log{kn})^7(p-q)^{r-1}n^{r-1}$. Then we have $\lVert M' \rVert_{row} \leq 0.1 \Delta$ and $\lVert M \rVert_{row} \leq \frac{1}{\log ^5 n} \Delta$.

By the condition that 
\[
\frac{p-q}{\sqrt{p}} \geq (C_0^{2}+1) k^{\frac{1}{2}}n^{-\frac{1}{2}+\frac{1}{2(r-1)}}(\log kn)^7
\]
where  $p,q \leq 0.75$ and $r\geq 3$,

we obtain that 
\[
\frac{p-q}{\sqrt{p(1-q)}} \geq (C_0^{2}+1) k^{\frac{1}{2}}n^{-\frac{1}{2}+\frac{1}{2(r-1)}}(\log kn)^7
\]

Thus, we can get that \[\lVert R^r \rVert_{row} \leq  0.1\Delta \]

By Lemma~\ref{lem:Lij}, we can see that for any $i\sim j$, with probability $1- O(n^{-1})$, $\lVert B^r_i - B^r_j\rVert \leq 0.3\Delta$. Otherwise, $i\nsim j$, $\lVert B^r_i - B^r_j\rVert \geq (2 - 0.3)\Delta = 1.7\Delta$. Therefore, with probability $1- O(n^{-1})$ we can use $\lVert  B^r_i - B^r_j \rVert$ determine whether any two nodes $v_i$ and $v_j$ are in the same cluster and Algorithm~\textsc{PowerIteration} can output $k$ clusters correctly. 
\end{proof}

\subsection{The MPC Algorithm}\label{isactive_mpc}

In this section, we show how to implement \textsc{PowerIteration} in the $s$-space  MPC model. The pseudocode is found in Algorithm~\ref{alg:mpc-2}. Given a matrix $A$, we use $A^{2r}_i$ to denote the $i$-th row of $A^{2r}$. We use $A^{2r}_{ \cdot j}$ to denote the $j$-th column of $A^{2r}$.

\begin{algorithm}[ht]
\caption{\textsc{MPC-PowerIteration}}
   \label{alg:mpc-2}
   \KwInput{ A SBM graph $G$, $\Delta$;}
\begin{algorithmic}[1]
 
\STATE Let $A$ be an adjacent matrix of $G$, $r$ be the parameter
\STATE Let $\Delta = C\sqrt{k}\sqrt{p(1-q)}(\log{kn})^7(p-q)^{r-1}n^{r-1}$, where $C\in \mathbb{R}_*^{+}$
\STATE $B = A - q\cdot J$ where $J = 1^{n\times n}$
\STATE{Let $i=1$ and $W=V$}
\STATE{Initially, all nodes in $W$ are active}
\WHILE{$\textsc{IsActive}(W)$ is true}
\STATE{Choose an arbitrary vertex $v\in W$ and send it to all machines}
\STATE{Label $v$ with $i$, i.e., $C_i=\{v\}$}

\FOR{each machine holding active vertex $u\in W$, we execute the following procedure in parallel
}
\STATE{\textsc{ComputeNorm($B,u,v,r$)}}
   \IF{$\lVert B_u^r - B_v^r\rVert \leq \Delta$}
   \STATE{Label $u$ with $i$}
   \ENDIF
   \ENDFOR
\STATE{Set nodes in $C_i$ inactive}
\STATE{$i=i+1$}
\ENDWHILE
\STATE Return all the sub-clusters $C_i$'s.
\end{algorithmic}
\end{algorithm}

We use $\textsc{IsActive}(W)$ to determine whether there are active nodes in $W$ or not. The details of implementation of $\textsc{IsActive}(W)$ in the $s$-space MPC model is as follows. We select one machine as the leader machine $M^*$. Let $\textsf{State}_M$ be the variable indicating that whether all nodes in the machine $M$ are all active or not. We use $\textsf{State}_M=1$ to indicate that all nodes in $M$ are active; $0$, otherwise. If $\textsf{State}_{M^*}=1$, the leader machine sends messages of $1$ to other machines. When a machine $M$ receives message, if $\textsf{State}_M=1$, then it forwards a message of $1$ to other machines; Otherwise, the machine $M$ sends a message of $0$ back to the sender. Therefore, after $O(\log_s n)$ rounds, the leader machine will know whether all nodes in $W$ are active or not.

We then use another subroutine $\textsc{ComputeNorm}(B,i,j,r)$ to compute $\lVert B^r_i - B^r_j \rVert$. Notice that we can't directly calculate matrix multiplication, which will take lots of rounds, we notice some good properties of $\lVert B^r_i - B^r_j \rVert$ and have the following theorem.

\begin{restatable}{theorem}{mpcRvivj}\label{thm:mpcRvivj} 
For a fixed $i$ and $j\in[n]$ and any integer $r<O(\log n)$, the subrountine $\textsc{ComputeNorm}(B,i,j,r)$ for computing $\lVert B_i^r-B_j^r\rVert$ can be implemented in $ O(r\log_s n)$ rounds where each machine has memory $s=\Omega(\log n)$.
\end{restatable}


Now we give the ideas of $\textsc{ComputeNorm}(B,i,j,r)$. We find that the expansion of $\lVert B^r_i - B^r_j \rVert$ has good properties such that we only need to compute the key terms for these $O(r^2)$ terms and the coefficients have good combinatorial explanations. Then we can use graph algorithms to calculate the results. First we note that 
\begin{eqnarray*}
	&&\vert\vert B_i^r-B_j^r\vert\vert^2=\vert\vert(\mathbf{1}_i^T- \mathbf{1}_j^T)(A-qJ)^{r}\vert\vert^2=(\mathbf{1}_i^T-\mathbf{1}_j^T)(A-qJ)^{2r}(\mathbf{1}_i-\mathbf{1}_j)\\
	~&=&\mathbf{1}_i^T(A-qJ)^{2r}\mathbf{1}_i-\mathbf{1}_i^T(A-qJ)^{2r}\mathbf{1}_j-\mathbf{1}_j^T(A-qJ)^{2r}\mathbf{1}_i+\mathbf{1}_j^T(A-qJ)^{2r}\mathbf{1}_j,
\end{eqnarray*}
so we only need to calculate $\mathbf{1}_x^T(A-qJ)^{2r}\mathbf{1}_y$ where $x,y\in \{i,j\}$.

Now we have the following lemma about the expanded formula.

\begin{restatable}{lemma}{expandlemma}\label{expanded_formula}
We have
$
    \mathbf{1}^T_x(A-qJ)^{2r}\mathbf{1}_y=  \mathbf{1}^T_xA^{2r}\mathbf{1}_y+\sum_{0\leq i_1\leq 2r-1}\sum_{0\leq i_t\leq 2r-1} X_{i_1,i_t}(A^{i_1}_x)J(A^{i_t}_{\cdot y}) 
$
where $X_{i_1,i_t}$ is coefficient only related to $n,q$ and $C_i$ and $C_i$ is the total number of different walks with length $i$ from $n$ vertices.
\end{restatable}

\begin{proof}
Let us expand $\lVert B_i^r-B_j^r\rVert^2$ first. 
\begin{eqnarray*}
	&&\vert\vert B_i^r-B_j^r\vert\vert^2\\
    ~&=&\vert\vert(\mathbf{1}_i^T- \mathbf{1}_j^T)(A-qJ)^{r}\vert\vert^2\\
	~&=&(\mathbf{1}_i^T-\mathbf{1}_j^T)(A-qJ)^{2r}(\mathbf{1}_i-\mathbf{1}_j)\\
	~&=&\mathbf{1}_i^T(A-qJ)^{2r}\mathbf{1}_i-\mathbf{1}_i^T(A-qJ)^{2r}\mathbf{1}_j-\mathbf{1}_j^T(A-qJ)^{2r}\mathbf{1}_i+\mathbf{1}_j^T(A-qJ)^{2r}\mathbf{1}_j
\end{eqnarray*}
So we only need to calculate $\mathbf{1}_x^T(A-qJ)^{2r}\mathbf{1}_y$ where $x,y\in \{i,j\}$.

We expand $\mathbf{1}^T_x(A-qJ)^{2r}\mathbf{1}_y$ into $2^{2r}$ terms. 
\[
\mathbf{1}^T_x(A-qJ)^{2r}\mathbf{1}_y = \mathbf{1}^T_xA^{2r}\mathbf{1}_y + \cdots + (-q)^{\sum_{l=1}^{t-1}j_{l}}\mathbf{1}_x^TA^{i_1}J^{j_1}\dots J^{j_{t-1}}A^{i_t}\mathbf{1}_y^T + \cdots + (-q)^{2r}\mathbf{1}^T_x J^{2r}\mathbf{1}_y
\]

 To avoid redundancy, we here only show how to deal with the key term $\mathrm{K}_t$, i.e., the term $(-q)^{\sum_{l=1}^{t-1}j_{l}}\mathbf{1}_x^TA^{i_1}J^{j_1}\dots J^{j_{t-1}}A^{i_t}\mathbf{1}_y^T$ of $\mathbf{1}_x^T(A-qJ)^{2r}\mathbf{1}_y$. For other terms like $A^{i_1}J,A^2r,JA^{i_t}$ which can be considered substitutions of this term and we have similar conversions and almost the same graph algorithms to deal with them.

Let $\textsc{c}^i_j$ denote the number of walks from the vertex $v_j$ to all vertices with length $i$. Then, we define $C_i = \Sigma_{j\in V} \textsc{c}^i_j$ which means the total number of walks with length $i$ from $n$ vertices to $n$ vertices. Then the key terms
\begin{eqnarray}
	&&(-q)^{\sum_{l=1}^{t-1}j_{l}}\mathbf{1}_x^TA^{i_1}J^{j_1}\dots J^{j_{t-1}}A^{i_t}\mathbf{1}_y \nonumber\\
	~&=& (-q)^{\sum_{l=1}^{t-1}j_{l}}n^{\sum_{l=1}^{t-1}(j_l-1)}\mathbf{1}_x^TA^{i_1}J\dots JA^{i_t}\mathbf{1}_y\nonumber\\
	~&=& (-q)^{\sum_{l=1}^{t-1}j_l}n^{-(t-1)+\sum_{l=1}^{t-1}(j_l-1)}\mathbf{1}_x^TA^{i_1}J(JA^{i_2}J)\dots (JA^{i_{t-1}}J)JA^{i_t}\mathbf{1}_y\nonumber\\
	~&=&(-q)^{\sum_{l=1}^{t-1}j_l}n^{-(t-1)+\sum_{l=1}^{t-1}(j_l-1)}\mathbf{1}_x^TA^{i_1}J(C_{i_2}J)\dots (C_{i_{t-1}}J)JA^{i_t}\mathbf{1}_y \nonumber\\
	~&=&(-q)^{\sum_{l=1}^{t-1}j_l}n^{1+\sum_{l=1}^{t-1}(j_l-1)}\left(\prod_{l=2}^{t-1}C_{i_l}\right) (\mathbf{1}_x^TA^{i_1})J(A^{i_t}\mathbf{1}_y) \nonumber\\
	~&=&(-q)^{\sum_{l=1}^{t-1}j_l}n^{1+\sum_{l=1}^{t-1}(j_l-1)}\left(\prod_{l=2}^{t-1}C_{i_l}\right) (A^{i_1}_x J A^{i_t}_{\cdot y}) \nonumber\\
	\nonumber
\end{eqnarray}
 Here we use $J^i=n^{i-1}J$ and $JA^{i}J=C_iJ$. 

 Notice that for the computation of $(A^{i_1}_x J A^{i_t}_{\cdot y})$, we only need to calculate the multiplication of the sum of elements in $A^{i_1}$ and the sum of elements in $A^{i_t}$.

Therefore,
\begin{eqnarray*}
    \mathbf{1}^T_x(A-qJ)^{2r}\mathbf{1}_y=  \mathbf{1}^T_xA^{2r}\mathbf{1}_y+\sum_{0\leq i_1\leq (2r-1)}\sum_{0\leq i_t\leq (2r-1)} X_{i_1,i_t}(A^{i_1}_x)J(A^{i_t}_{\cdot y}) 
\end{eqnarray*}
Here $X_{i_1,i_t}$ is the sum of coefficients of all terms contain $(A^{i_1}_x)J(A^{i_t}_{\cdot y})$. 
\end{proof}

Since $r=O(\log n)$, there are $\poly(\log n)$ terms in the right hand side. So we can store all coefficients in $\tilde O(n)$ space. Notice that $\mathbf{1}^T_xA^{2r}\mathbf{1}_y$ for all $y\in[n]$ is the $i^{th}$ row of $A^{2r}$, i.e., $A^{2r}_x$. To compute $\mathbf{1}^T_x(A-qJ)^{2r}\mathbf{1}_y$, we split it into computing $C_i$, $A^{i_1}_x J A^{i_t}_{\cdot y}$, and $A^{2r}_x$.

\bigskip

\noindent \textbf{Compute $C_i$ and $A^{i_1}_x J A^{i_t}_{\cdot y}$.}~~We show how to compute the value of any $C_i$ and $A^{i_1}_x J A^{i_t}_{\cdot y}$. Let $\vec{\mathbf{j}}$ be all ones vector, which is a column of $J$.
To compute $A^i_x \vec{\mathbf{j}}$, we propose a simple algorithm, i.e., Algorithm~\ref{alg:4} that is described as follows. 


\begin{algorithm}[h]
 \SetAlgoLined
	\caption{\textsc{AixSum}($G(n),r$): Calculating $A^{i}_x \vec{\mathbf{j}}$ for all $x\in[n], i\in[2r]$}
	\label{alg:4}
 \KwInput{A graph $G(n)$,$r$}
	\begin{algorithmic}[1]
		\FOR{each node $u$ in $G$}
		\STATE{$A_{u,0}=1$}
		\ENDFOR
		\STATE{let $i=1$}
		\WHILE{$i\leq 2r$}
		\FOR{each node $u$ in $G$}
		\STATE{$sum_u=0$}
		\FOR{each neighbor $v$ of $u$}
		\STATE{$sum_u+=A_{v,i-1}$}
		\ENDFOR
		\STATE{$A_{v,i}=sum_u$}
		\ENDFOR
		\STATE{$i=i+1$}
		\ENDWHILE
		\STATE Return $A$, $A_{x,i}$ is $A^i_x \vec{\mathbf{j}}$
	\end{algorithmic}
\end{algorithm}

\begin{restatable}{lemma}{computeCi}\label{lem:computeCi}
For all $x\in[n], i\in[2r]$, Algorithm~\ref{alg:4} outputs $A^{i}_x \vec{\mathbf{j}}$.
\end{restatable}

\begin{proof}
The pseudo-code of our algorithm is  Algorithm~\ref{alg:4}. Now, we show the correctness of Algorithm~\ref{alg:4} by induction. We first define that at the end of the $i$-th phase, for each node $u$, the value of $A_{u,i}$ is the number of walks ending at the node $u\in V$ with step size $i$. In the first phase, for a node $u\in V$, $A_{u,1} = |N(u)|$. Obviously, $|N(u)|$ is the number of walks ending at $u$ with step size one. Suppose it is true for the $k$-th phase, for a node $u\in V$, $A_{u,k}$ is the number of walks ending at $u$ with step size $k$. In the $k+1$ phase, as each node will receive values from its neighbors, $A_{u,k+1} = \Sigma_{v \in N(u)} A_{v,k}$ and we can see that the value of $A_{u,k+1}$ is the number of walks ending at $v$ after $k+1$ steps. Inductively, we prove our claim. Next, we prove that $A^i_x \vec{\mathbf{j}} = A_{x,i}$. We can see that $A^i_x \vec{\mathbf{j}}$ means that the number of walks from $x$ after walks with step size $i$ and $A_{x,i}$ means that the number of walks ends at $x$ after walks with step size $i$. Notice that here the walks mean all possible walks. Since we consider all possibilities, after walks with step size $i$, the number of walks ending at $x$ is the number of walks starting from $x$. Therefore,  $A^i_x \vec{\mathbf{j}} = A_{x,i}$.
\end{proof}

To compute $C_i$ for any $i\in [n]$, we only need to sum up $A^{i}_x \vec{\mathbf{j}}$ for all $x,i\in[n]$.

Now, let us see how to implement Algorithm~\ref{alg:4} in the $s$-space MPC model. Notice that in default, we use the fact that in the $s$-space MPC model, each vertex can visit its neighbors in $O(\log_s n)$ rounds where each machine has memory of $O(n^{\delta})$ by Lemma~\ref{lem:neigh-mpc}.

\begin{restatable}{lemma}{computeAix}\label{lem:computeAixmpc}
    In $s$-space MPC model, for all $i\in [2r]$ and $x\in[n]$, there exists an algorithm that can compute all $A^{i}_x \vec{\mathbf{j}}$ in $O(r\log_s n)$ rounds,where each machine has memory $s=\Omega(\log_s n)$.
\end{restatable}
\begin{proof}
We transform Algorithm~\ref{alg:4} as follows. In each machine $M_i$, each node $u\in M_i$ has a counter variable~$u_\phi$. ~$u_\phi$ here maintains the value of $A_{u,i}$ in round $i$. Initially~$u_\phi = 1$ for all nodes in $V$. Then, each node sends $u_\phi$ to its neighbors. Notice that in MPC model, each machine stores edges and vertices. Therefore, the process happens within each single machine. We use $u_\phi(i)$ to denote the part of value of $u_\phi$ in the $i$-th machine ($i\in[p]$ where $p$ is the number of machines). Then, we need to accumulate all partial results, i.e.,~$u_\phi(i)$ ($i\in[p]$) into a complete~$u_\phi$. It means that we need to find all neighbors of any $u\in V$. As we mentioned before, in sublinear MPC model, we can finish it in $O(\log_s n)$ rounds. Also, we can return $u_\phi$ to each~$u_\phi(i)$ ($i\in[p]$) within~$O(\log_s n)$ rounds in reverse. Therefore, it takes $O(\log_s n)$ rounds to finish one phase. As there are $r$ phases together, the round complexity is $O(r\log_s n)$.
\end{proof}

Notice that~$A^{i_1}_x JA^{i_t}_{\cdot y}=A_x^{i_1}\vec{\mathbf{j}}(A_y^{i_t}\vec{\mathbf{j}})$, $C_i=\sum_{x\in [n]}A_x^i\vec{\mathbf{j}}$ and we have computed $A^{i}_x \vec{\mathbf{j}}$ for all $i\in[2r]$ and $x\in[n]$, we can obtain $A^{i_1}_x J A^{i_t}_{\cdot y}$ in constant rounds. So the main round complexity is only about the calculation of $A^{i}_x \vec{\mathbf{j}}$ and we have the following corollary.

\begin{restatable}{corollary}{computeCimpc}\label{cor:computeCimpc}
    In $s$-space MPC model, there exists an algorithm that can compute $A^{i_1}_x J A^{i_t}_{\cdot y}$ and $C_i$ in~$O(r\log_s n)$ rounds, where each machine has memory $s=\Omega(\log n)$.
\end{restatable}

\noindent \textbf{Compute $A^{2r}_x$ .}
Note that each entry $a^{2r}_{x,y}$ in $A^{2r}_x$ is exactly the number of walks with step size $r$ from $v_x$ to $v_y$. The naive algorithm of computing $A^{2r}$ is to compute the matrix, but it is resources-consuming. Another idea is to compute~$a^{2r}_{x,y}$ for any $x$ and $y$, respectively. If each machine can store all vertices within radius $r$, then we can directly compute all~$a^{2r}_{x,y}$ ($x,y\in[n]$). Now, we consider the $s$-space MPC model, i.e., single machines cannot store vertices within radius $r$.

We use procedure $\textsc{ComputeArx}$ to solve this problem of computing $A^{2r}_x$. Take $a^{2r}_{x,y}$ for example. The goal is to compute the number of walks from $v_x$ to $v_y$ after walks with step size $r$. \\

\noindent\textbf{Step 1:} Initially, we compute the value the size of $N(v_x)$ by sending messages to all machines. Let $|N_i(v_x)|$ be the returned value from the $i$-th machine that represents the number of neighbors of $v_x$ in the $i$-th machine. We know that $\Sigma{|N_i(v_x)|} = |N(v_x)|$. \\

\noindent\textbf{Step 2:} Next, we assign $|N(v_x)|$ tokens to the vertex $v_x$. Each token has a time variable $\Phi$ that represents the value of the current round. Each token also has a value $\mathcal{V}$. Before the start of algorithm, the value of $\mathcal{V}$ for each token is 1. After the first round, the value of $\mathcal{V}$ of tokens holding in a vertex $v$ is equal to the sum of values of tokens that $v$ receives. We use $\mathcal{V}(v)$ to denote the value $\mathcal{V}$ of tokens hold by the vertex $v$. Therefore, the values of $\mathcal{V}$ for tokens are updated for each round. We let $v_x$ send $|N_i(v_x)|$ tokens to corresponding machines respectively. Each neighbor of $v_x$ receive one token and the vertex receiving a token increases the value of $\Phi$ by one.

We repeat step~1 and step~2 until there exist tokens with $\Phi = r$. Then, we count the number of tokens at the vertex $v_j$. 

\begin{restatable}{lemma}{mpcAr}
Let $a^{2r}_{x,y}$ be the number of walks from the vertex $x$ to the vertex $y$ after walks with step size $2r$. There exists a procedure $\textsc{ComputeArx}$ that can find $a^{2r}_{x,y}$ after $O(r\log_s n)$ rounds where each machine has memory $s=\Omega(\log n)$.
\end{restatable}
\begin{proof}  
First, we prove the correctness of procedure $\textsc{ComputeArx}$. We use $W_s(v)$ to denote the set of nodes that can be reached by the vertex $v$ after walks with step size $s$. The proof is by induction on $s$. The base $s = 1$ is immediate. Now, we assume that after $s-1$ rounds,  the number of walks with step size $s-1$ from $v_i$ is  $ a^{s-1}_{i,j}$. In the $s$-th round, The number of $s$ walks from $v_i$ is sum of walks from those nodes which are in $W_{s-1}(v)$ walk one step. Therefore, $\Sigma_{u\in N(v_j)} \mathcal{V}(u) = \Sigma_{u\in N(v_j)} a^{s-1}_{i,u}=a^{s}_{i,j}$. Inductively, we prove the correctness. 

Now, let us see the round complexity. In each round, each node has to access its neighbors, which can be done in $O(\log_s n)$ rounds. Since we need to repeat accessing neighbors for $r$ times due to $r$ walks, the total number of round complexity is $O(r\log_s n)$. 
\end{proof}

 By taking the union of different vertices, we can get the following corollary. 
\begin{restatable}{corollary}{Armpc}\label{cor:Armpc}
Let $A^{2r}_{i}$ be set of the numbers of walks from the vertex $i$ to the vertex $j$ where $j\in[1,n]$ after walks with step size $r$. The procedure $\textsc{ComputeArx}$ can find $A^{2r}_{i}$ after $O(r\log_s n)$ rounds where each machine has memory $s=\Omega(\log n)$.
\end{restatable}

\noindent \textbf{Compute $\lVert B_i^r-B_j^r\rVert$.}~~After obtaining $A^i_x$ and $A^i_{\cdot y}$, the value of $A^i_x J A^i_{\cdot y}$ is the multiplication of the sums of terms in each of two vectors. And the coefficient of each term is the multiplication of $C_i$, $n$ and $q$. 
Now, we can prove Theorem~\ref{thm:mpcRvivj}.
\begin{proof}[\textbf{Proof of Theorem~\ref{thm:mpcRvivj}}]
By Lemma~\ref{expanded_formula}, $\mathbf{1}^T_x (A-qJ)^{2r}\mathbf{1}_y$ which consists of at most $ O(r^2)$ terms with $C_i$, $A^{i_1}_x J A^{i_t}_{\cdot y}$, and $A^{2r}_x$. Let us see the round complexity of computing it. We take the round complexity of computing one key term as an example. We only need to look at the round complexity of computing~$\left(\prod_{l=2}^{t-1}C_{i_l}\right) (A^{i_1}_x)J(A^{i_t}_{\cdot y})$. By Corollary~\ref{cor:computeCimpc}, we need $O(i\log_s n)$ rounds to compute any $C_i$ where $i\in[2r]$. We can finish calculating~$\left(\prod_{l=2}^{t-1}C_{i_l}\right)$ in $O(r\log_s n)$ rounds. By Lemma \ref{lem:computeAixmpc}, we can obtain the result of $(A^{i_1}_x)J(A^{i_t}_{\cdot y})$ within~$O(r\log_s n)$ rounds.  Notice that there is a special term $\mathbf{1}^{T}_x A^{2r} \mathbf{1}^T_y$. By Corollary \ref{cor:Armpc}, we can find it within~$O(r\log_s n)$ rounds. There are some other similar computations, which also take $O(r\log_s n)$ rounds. Recall that there are $O(r^2)$ terms, we deal with it by copying the whole graph for $\text{poly}(\log n)$ times and then put these results together. Therefore, it takes $ O(r\log_s n)$ rounds to calculate~$\mathbf{1}^T_x (A-qJ)^{2r}\mathbf{1}_y$. Therefore, we need  $ O(r\log_s n)$ rounds to finish the calculation of $\lVert B_i^r-B_j^r\rVert$.
\end{proof}

By Theorem~\ref{thm:mpcRvivj}, we can finish the proof of Theorem~\ref{thm:secondmpc}.

\begin{proof}[\textbf{Proof of Theorem~\ref{thm:secondmpc}}]
  The main idea of \textsc{MPC-PowerIteration}(Algorithm~\ref{alg:mpc-2}) is to fix a node $v_i$ first and calculate~$\lVert B^r_i - B^r_j \rVert$ for any other node $v_j$ in the same cluster to obtain all nodes in the same cluster. We need to store all simplified coefficients  in each round which uses $O(nr^2)$ space. For other operations in the algorithms, $O(m)$ space is enough. So the total space complexity is $\tilde O(m+nr^2) = \tilde O(m)$.
  
Recall that we  implement $\textsc{IsActive}(W)$ in $O(\log_s n)$ rounds. By Theorem~\ref{thm:mpcRvivj}, we can finish $\lVert B^r_i - B^r_j \rVert$ within $O(r\log_s n)$ rounds.  Therefore, we need $O(r\log_s n)$ rounds to find a cluster and all its nodes. There are $k$ hidden clusters and we execute the above procedure sequentially, so the round complexity is $O(kr\log_s n)$.
\end{proof}

Now we show how to use more space to trade off round complexity and give  the proof of Theorem~\ref{thm:mpc-2}. 

\begin{proof}[\textbf{Proof of Theorem~\ref{thm:mpc-2}}]
Recall that in \textsc{MPC-PowerIteration}(Algorithm~\ref{alg:mpc-2}), we sequentially find $k$ clusters, that is the reason why there is a factor $k$ in the round complexity. Now, we execute the process in parallel. First, we randomly sample a set $S_k$ of $\Theta(k\log n)$ nodes. With high probability, for each hidden cluster, we sample $\Theta(\log n)$ nodes in $S_k$. Then, for each node $u \in S_k$, we execute $\textsc{ComputeMatrix}(B,u,v,r)$ for each $v\in V$. If $\lVert B^r_u - B^r_v\rVert \leq \Delta$, we put $u$ and $v$ in the same cluster. The space for this step is $\tilde O(km+knr^2\log n)$. Then, we will have $k$ clusters with $\Theta(k\log n)$ labels. We remove duplicated labels by keeping the label with the minimum value among all received labels to get one label vertex for each cluster. Then by using these $k$ label vertices, we can use $\tilde O(km+knr^2\log n) = \tilde O(km)$ space to cluster all vertices. So, we can find all $k$ clusters in $O(r\log_s n)$ rounds with high probability. 
\end{proof}

\bibliographystyle{plainurl}
\bibliography{ref}

\appendix

\input{Main-Appdendix}

\end{document}

%% file: Main-Appdendix.tex
\section{Some Previous MPC Operations}\label{sec:appendix}
In this section, we give definitions and theorems of some previous basic and frequently used operations in MPC model.

\subsection*{Indexing}
In the index problem, a set $S = \{x_1,x_2,\ldots,x_n\}$ of $n$ items are stored in machines. The output is 
\[
S' = \{(x,y)| x\in S, y-1 = n(x) \}
\]
where $n(x)$ is the number of items before $x$. 
\subsection*{Prefix Sum}
In the prefix sum problem, a set $S = \{(x_1,y_1),(x_2,y_2),\ldots,(x_n,y_n)\}$ of $n$ pairs are stored on machines. The output is 

\[ S' = \bigg\{(x,y') |(x,y)\in S, y'-y = \sum_{(x',y')<(x,y)} y' \bigg\},\]

where $(x',y') < (x,y)$ means that $(x',y')$ is held by an input machine with a smaller index or $(x',y')$ and $(x,y)$ are in the same machine but $(x',y')$ has a smaller local memory address.

\subsection*{Copies of Sets}

Suppose we have $k$ sets $S_1,S_2,\ldots,S_k$ stored in machines. Let $s_1,s_2,\ldots,s_k \in \mathbb{Z}$. Each machine knows the value of $s_i$ if holding an element $x\in S_i$. The goal is to create sets $S_{1,1},S_{1,2},\ldots,S_{1,s_1},\ldots,S_{k,s_k}$ in machines where $S_{i,j}$ is the $j$-th copy of $S_i$.

\begin{theorem}[\cite{andoni2018parallel,GoodrichSZ11}]\label{them:basics}
    Indexing/prefix sum/copies of sets can be solved in $O(\log_s n)$ rounds in the $s$-space MPC model.
\end{theorem}